\documentclass[11pt]{article}
\usepackage{amssymb}
\usepackage{colortbl}
\usepackage{amsfonts,amsmath, longtable}

\usepackage{comment}

        \def\theequation{\thesection.\arabic{equation}}

\newcommand{\ti}[1]{\tilde{#1}}

\newcommand{\mF}{{\mathcal F}}

\newcommand{\mH}{{\mathcal H}}

\newcommand{\mD}{{\mathcal D}}

\newcommand{\vf}{\varphi}
\newcommand{\al}{\alpha}
\newcommand{\be}{\beta}
\newcommand{\ga}{\gamma}
\newcommand{\om}{\omega}
\newcommand{\vth}{\vartheta}

\newcommand{\MatM}{ {\rm Mat}(M,\mathbb C) }

\newcommand{\mC}{\mathbb C}
\newcommand{\mZ}{\mathbb Z}

\newtheorem{predl}{Proposition}[section]
\newtheorem{lemma}{Lemma}[section]

\newenvironment{proof}{\par\noindent{\bf Proof.}}{\hfill$\scriptstyle\blacksquare$}

\def\beq{\begin{equation}}
\def\eq{\end{equation}}
\def\p{\partial}

\newtheorem{theor}{Theorem}

\newcommand{\mats}[4]{\left(\begin{array}{cc}{#1}&{#2}\\ {#3}&{#4}
\end{array}\right)}

\def\res{\mathop{\hbox{Res}}\limits}

\begin{document}

\setcounter{page}{1}

\begin{center}

\

\vspace{-0mm}




{\Large{\bf Elliptic  generalization of integrable q-deformed   }}

\vspace{3mm}

{\Large{\bf anisotropic Haldane-Shastry long-range spin chain }}

 \vspace{15mm}

 {\Large {M. Matushko}}
\qquad\quad\quad
 {\Large {A. Zotov}}

   \vspace{5mm}

{\em Steklov Mathematical Institute of Russian
Academy of Sciences,\\ Gubkina str. 8, 119991, Moscow, Russia}

   \vspace{3mm}

 {\small\rm {e-mails: matushko@mi-ras.ru, zotov@mi-ras.ru}}

\end{center}

\vspace{0mm}

\begin{abstract}
 We describe integrable elliptic q-deformed anisotropic long-range spin chain. The derivation is based
 on our recent construction for commuting anisotropic elliptic spin Ruijsenaars-Macdonald operators.
 We prove that the Polychronakos freezing trick can be applied to these operators, thus providing
 the commuting set of Hamiltonians for long-range spin chain constructed by means of the elliptic Baxter-Belavin
 ${\rm GL}_M$ $R$-matrix. Namely, we show that the freezing trick is reduced to a set of elliptic function identities, which are then proved.
 These identities can be treated as conditions for equilibrium position
 in the underlying classical spinless Ruijsenaars-Schneider model.
 Trigonometric degenerations are studied as well. 
  For example, in $M=2$ case
 our construction provides q-deformation for anisotropic XXZ Haldane-Shastry model.
 The standard Haldane-Shastry model and its Uglov's q-deformation based on
  ${\rm U}_q({\widehat {\rm  gl}_M})$ XXZ $R$-matrix are included into consideration by separate verification.
 %

\end{abstract}

\newpage



\section{Introduction: long-range spin chains and many-body systems}
\setcounter{equation}{0}

\subsection{Isotropic models}
The first example of integrable long-range spin chain was introduced by F.D.M. Haldane  and B.S. Shastry  \cite{HS1,HS2}. The Hamiltonian
 \beq\label{s01}
 \begin{array}{c}
  \displaystyle{
 H^{\rm{HS}}=
  \frac{1}{2}\sum\limits_{i\neq j}^N \frac{1-P_{ij}}{\sin^2(\pi(x_i-x_j))}
  }
 \end{array}
\eq
 describes pairwise interaction of $N$ spins being attached to
 equidistant points on a circle: $x_k=k/N$, $k=1,...,N$. Here $P_{ij}$ are the permutation operators (or spin exchange operators), which act on the Hilbert space
 $\mH=(\mC^M)^{\otimes N}$ by permuting $i$-th and $j$-th tensor components. For the ${\rm su}_2$ case $M=2$
 and
 \beq\label{s02}
 \begin{array}{c}
  \displaystyle{
  P_{ij}=\frac{1}{2}\sum\limits_{a=0}^3 \sigma_a^{(i)}\sigma_a^{(j)}\,,\qquad
  \sigma_a^{(i)}=\underbrace{1_2\otimes...1_2\otimes\sigma_a\otimes 1_2...\otimes 1_2  }_{\sigma_a\ {\rm is\ on\ the\ i-th\ place}}\in{\rm Mat}_{2^N}\,,
  }
 \end{array}
\eq
 where $\sigma_a$ are the Pauli matrices (spin $1/2$ operators), and $1_2=\sigma_0$ is the identity $2\times 2$ matrix. For arbitrary $M$ the permutation operator is given in (\ref{P12}). The Hamiltonian (\ref{s01})
 resembles\footnote{Both Hamiltonians are sums of permutation operators (with some
coefficients).} the one for the Heisenberg magnet $H^{Heis}=\sum_k P_{k,k+1}$ but the interaction is between all spins (not only neighbours). For this reason  (\ref{s01}) is called the long-range spin chain. The Heisenberg magnet
 is isotropic (or XXX) model since all $\sigma_a$ ($a=1,2,3$) enter the spin-exchange operators on equal footing. Due to the fact that (\ref{s01}) also depends on non-deformed permutation
 operators, it is natural to keep the terminology and call such models isotropic (or XXX).

 While the Heisenberg magnet can be described (and solved) by the quantum inverse scattering method\footnote{The models coming from RTT relations may describe non-local interaction as well. Examples are inhomogeneous spin chains.} (based on RTT-relations, commuting transfer-matrices and the Bethe ansatz technique), the origin
 of the Haldane-Shastry type models is different.
The Hamiltonian (\ref{s01}) can be included into a wide class of integrable models.
Integrable long-range ${\rm gl}_M$ spin chains on $N$ sites of the Haldane-Shastry type are defined by Hamiltonians of the form:
 \beq\label{s03}
 \begin{array}{c}
  \displaystyle{
 H^{\rm{XXX}}=
  \frac{g}{2}\sum\limits_{i\neq j}^N P_{ij}\,U(x_i-x_j)\in{\rm End}(\mH)\,,
  }
 \end{array}
\eq
 where $U(x)$ is a certain function, $g\in\mC$ is a constant parameter and
  $x_1,...,x_N$ is a special set of points. For example, the elliptic Weierstrass $\wp$-function $U(x)=\wp(x)$
  also provides integrable model, known as the Inozemtsev long-range spin chain \cite{Inoz}.
   Possible choices of  the function $U(x)$ come from
  the similarity \cite{HS2} between the Haldane-Shastry model (\ref{s01}) and
  the spin generalizations \cite{GH} of the Calogero-Moser-Sutherland models \cite{Calogero1}, which are defined
  by the Hamiltonian of the form:
 \beq\label{s04}
 \begin{array}{c}
  \displaystyle{
 H^{\rm{spin\, CM}}=\frac{\rm Id}2\sum\limits_{k=1}^N  \eta^2\p_{z_k}^2+
  \frac{1}{2}\sum\limits_{i\neq j}^N (\hbar^2{\rm Id}-\eta\hbar P_{ij})\,U(z_i-z_j)\,,\qquad {\rm Id}=1_{({\mC^M)}^{\otimes N}}\,,
  }
 \end{array}
\eq
 acting on $({\mC^M)}^{\otimes N}$-valued functions of $z_1,...,z_N$, where $\eta$ is the Planck constant and $\hbar$ is a coupling constant\footnote{It seems misleading to use notation
  $\hbar$ for the coupling constant, but $\hbar$ can be also considered as the Planck constant in
  the $R$-matrix classical limit (\ref{r052}). At the same time $\hbar$ plays the role of the coupling constant
  in (spin) many-body systems.}. The procedure relating (\ref{s04}) and (\ref{s03}) is called the Polychronakos freezing trick \cite{Polych1,Polych0}.
 Loosely speaking, it states that one should remove the terms with differential operators from the spin Calogero-Moser Hamiltonians and fix the positions of particles as equilibrium positions of the underlying
 spinless classical model.\footnote{Main observation by A.P. Polychronakos \cite{Polych1,Polych0} was that the quantum fluctuations of coordinates around their classical equilibrium positions decrease with growing of coupling constant, and thus contribute subleading terms in the spin-dependent part of the Hamiltonian. In the infinite coupling limit the full partition function factorizes into the partition function of the kinematic degrees of freedom times the partition function of the spin system. IN this way the exact spectrum of the spin system can be found.} In this way (\ref{s04}) is mapped to (\ref{s03}) \cite{CoSa,IS}. For example, in the Haldane-Shastry and Inozemtsev models
  $z_k=x_k$, and these positions are indeed equilibrium positions in the underlying classical mechanics. The Hamiltonian of the spinless Calogero-Moser model \cite{Calogero2} takes the form
 \beq\label{s05}
 \begin{array}{c}
  \displaystyle{
 h^{\rm CM}=\frac{1}2\sum\limits_{k=1}^N v_k^2-
  \frac{\nu^2}{2}\sum\limits_{i\neq j}^N U(z_i-z_j)\,,
  }
 \end{array}
\eq
 where $v_k$ are momenta (with the canonical Poisson brackets $\{v_i,z_j\}=\delta_{ij}$) and $\nu$ is the classical coupling constant.
 The set $z_k=x_k$, $k=1...N$ solves the system of equations
 \beq\label{s06}
 \begin{array}{c}
  \displaystyle{
 {\dot v}_k={\ddot z}_k=\nu^2\sum\limits_{j:j\neq i}U'(z_i-z_j)=0\,,\quad i=1,...,N
  }
 \end{array}
\eq
 for $U(x)=1/\sin^2(\pi x)$ and $U(x)=\wp(x)$. For the rational Calogero-Moser potential $U(x)=1/x^2+ x^2/2$ (with oscillator terms) the equilibrium positions are more complicated. The positions are given by zeros of Hermite polynomials \cite{Calogero3,OSa}. In this case one gets the Polychronakos-Frahm chain \cite{Polych1,Frahm}. At the same time it should be mentioned that the freezing trick is not a precise statement but rather
 a general recipe, which is needed to be proved for any concrete model.

For spin chains coming from the quantum inverse scattering method the commuting set of Hamiltonians appear from
RTT relations and (higher) commuting transfer-matrices. In contrast to this simple construction, the higher commuting Hamiltonians for the long-range spin chains are calculated not so straightforwardly. However, this problem can be solved
for rational and trigonometric models in several different ways. The first possibility is to use the Dunkl-Polychronakos type approach \cite{Polych1,FM}, the second way is to use hidden Yangian symmetry \cite{HHTBP,TH} (the spin chain
 Hamiltonians appear as the center elements of the Yangian algebra), and the third way is to
 use the quantum Lax pairs \cite{IS} and calculate the so-called total sum of powers of quantum Lax matrix.
 All these approaches however are not applicable to elliptic case. For the Inozemtsev chain the higher Hamiltonians
 were proposed in \cite{Inoz2}. These Hamiltonians were proved to commute with (\ref{s03}) but mutual commutativity still remains to be an open problem (see also \cite{FGL1}).

Like the Calogero-Moser-Sutherland models, the long-range spin chains can be extended
to other root systems (not only of $A_{N-1}$ type) \cite{BPS,FGL2} and to supersymmetric case \cite{BMB}.
Also, similarly to integrable many-body systems the long-range spin chains have applications in different
areas of theoretical and mathematical physics, see e.g. \cite{Serban}.

\subsection{Anisotropic models} The above mentioned Heisenberg magnet
 has integrable anisotropic generalizations (Landau-Lifshitz type models) given by $H^{\rm LL}=\sum_k\sum_a J_a\sigma_a^{(k)}\sigma_a^{(k+1)}$. When $J_1$, $J_2$ and $J_3$ are distinct the model is called XYZ. It is described by elliptic Baxter's $R$-matrix. The case $J_1=J_2$ is called XXZ (partially anisotropic), and the underlying $R$-matrix is trigonometric.

 The first example of anisotropic long-range spin chain was introduced by D. Uglov \cite{Uglov} (see also \cite{Lam,LPS} and \cite{HHTBP,TH}), and it is also a q-deformed
 model, which we discuss below. But before proceeding to
 q-deformed models let us describe a natural extension of the previously discussed long-range chains to anisotropic case.
 A general form for anisotropic ${\rm gl}_M$ model is as follows:
 \beq\label{s07}
 \begin{array}{c}
  \displaystyle{
 H^{\rm{anis}}=
  \frac{g}{2}\sum\limits_{i\neq j}^N
  \sum\limits_{a,b,c,d=1}^M e_{ab}^{(i)}e_{cd}^{(j)}
  U_{ab,cd}(x_i-x_j)\in{\rm End}(\mH)\,,
  }
 \end{array}
\eq
 where $e_{ab}^{(i)}$ is the standard matrix basis matrix
 $e_{ab}\in{\rm Mat}_M$ in the $i$-th tensor component of $\mathcal H$.
 The  Hamiltonian (\ref{s07}) becomes isotropic in the case
 $U_{ab,cd}(x_i-x_j)=\delta_{ad}\delta_{bc}\,U(x_i-x_j)$. Then
 (\ref{s07}) reproduces (\ref{s03}). The corresponding
 anisotropic spin Calogero-Moser Hamiltonian is as follows:
 \beq\label{s08}
 \begin{array}{c}
  \displaystyle{
 H^{\rm{anis\, CM}}=\frac{\rm Id}2\sum\limits_{k=1}^N  \eta^2\p_{z_k}^2+
  \frac{g}{2}\sum\limits_{i\neq j}^N
  \sum\limits_{a,b,c,d=1}^M e_{ab}^{(i)}e_{cd}^{(j)}
  U_{ab,cd}(z_i-z_j)\,.
  }
 \end{array}
\eq
 The classical elliptic models of this type were introduced in \cite{Polych2}. Later these models were rediscovered in the Hitchin framework to integrable systems \cite{LZ,LOSZ}. They were called the models of interacting (integrable) tops since
 in the anisotropic case even one-site model is non-trivial. It is
 a multidimensional Euler-Arnold integrable top.
 General systems of interacting tops including the rational and trigonometric cases were described through $R$-matrix data in
   \cite{GZ}. In \cite{GZ} the quantum models were studied in the context of quantum Lax pairs, and some examples of
   systems related to classical root systems were proposed as well.

  Summarizing, in the anisotropic case we have a wide class of models of interacting tops instead of the spin Calogero-Moser models in isotropic case. It is then natural to define the corresponding anisotropic long-range spin chains through (\ref{s07}). In this way the elliptic XYZ model was suggested in \cite{SeZ}. Its trigonometric limit provides XXZ analog of (\ref{s01}). In ${\rm gl_2}$ case it is as follows \cite{SeZ}\footnote{The particle spin model with the Hamiltonian (\ref{s08}) and interaction (\ref{s09}) was introduced in \cite{HW}. See also (\ref{s0911}) for 7-vertex deformation of (\ref{s09}).}:
 \beq\label{s09}
 \begin{array}{c}
  \displaystyle{
 H^{\rm{XXZ}}=
  \frac{g}{2}\sum\limits_{i\neq j}^N
  \frac{\cos(\pi(x_i-x_j))(\sigma_1^{(i)}\sigma_1^{(j)}+
  \sigma_2^{(i)}\sigma_2^{(j)})+\sigma_3^{(i)}\sigma_3^{(j)}
  }{\sin^2(\pi(x_i-x_j))}\,.
  }
 \end{array}
\eq
 We call this model anisotropic (XXZ) Haldane-Shastry model.
  The problem here is to compute higher commuting Hamiltonians. An attempt was made in \cite{SeZ,Z18}, where the next non-trivial Hamiltonian was evaluated from the $R$-matrix valued Lax pairs, but the problem of finding higher Hamiltonians and proving their commutativity was not solved.

\subsection{q-deformed models} The term q-deformation is usually used, when a generalization based on some kind of quantum group structure is discussed. While integrable many-body systems of the Calogero-Moser type are constructed in terms of Lie algebras \cite{Calogero1,Calogero2}, their lift to the classical or quantum Lie group level corresponds to the Ruijsenaars-Schneider systems \cite{Ruij0,Ruij}. In order to apply the approach described above one should deal with the spin generalizations of the Ruijsenaars-Schneider models. This generalization was introduced in \cite{KrichZ} at the level of classical mechanics. However, the Hamiltonian description and its quantization in general (elliptic) case is still an open problem. At the same time much progress was achieved in studies of the trigonometric models \cite{AO}.

The first q-deformed long-range spin chain was proposed in \cite{Uglov} using a different approach motivated by studies of Hecke algebras and results of \cite{HHTBP,TH}.
The Uglov's construction was revisited and clarified recently in
\cite{Lam,LPS}. Main idea is as follows. Consider the set of (commuting) Macdonald operators
($k=1,\dots,N$):
\begin{equation}\label{Macd}
 D_k^{\rm Macd}=\sum\limits_{\substack{|I|=k}}A_I\prod_{i\in I}q^{-y_i\p_{y_i}}\,,\quad A_I=\left(\frac{2\pi\imath}{t-1}\right)^{k(N-k)}
 \prod\limits_{\substack{i\in I \\ j\notin I}}\frac{ty_j-y_i}{y_j-y_i}\,,
 \quad y_k=e^{2\pi\imath z_k}\,,
\end{equation}
which are also Hamiltonians of the quantum trigonometric Ruijsenaars-Schneider model \footnote{(\ref{Macd}) differs from the definition in \cite{Macd} by changing the parameters $q\to q^{-1}$ and $t\to t^{-1}$ and up to a constant factor.}. Uglov suggested a q-deformed version of the construction from \cite{HHTBP} (see also 
\cite{LPS}).
The result is that
the Macdonald operators (\ref{Macd}) admit anisotropic spin generalization of the form:
\begin{equation}\label{spMacd}
 \mD_k^{\rm XXZ}=\sum\limits_{\substack{|I|=k}}A_I\,{\bf R}_I \Big(\prod_{i\in I}q^{-y_i\p_{y_i}}\Big){\bf R}_I^{-1}\,,
\end{equation}
where ${\bf R}_I$ are certain products of ${\rm GL}_2$ XXZ $R$-matrices
(we give explicit form for ${\bf R}_I$ below).
Main statement is that $\mD_k^{\rm XXZ}$ also mutually commute. Then, using a kind of freezing trick, it was shown in \cite{Uglov} (following idea from \cite{HHTBP}) that a set of commuting long-range spin chain Hamiltonians can be deduced from
(\ref{spMacd}), where the positions $z_j$ are fixed as $x_j=j/N$.
These are equilibrium positions in the classical trigonometric Ruijsenaars-Schneider model \cite{Ruij3}, see also \cite{CoSa,RaSa}. In this
way the described above recipe for the freezing trick is performed as was recently shown by J. Lamers, V. Pasquier and D. Serban \cite{LPS}.

\subsection{Purpose and plan of the paper}

\paragraph{Purpose of the paper} is to introduce commuting set of
Hamiltonians for long-range spin chain based on the elliptic ${\rm GL}_M$ $R$-matrix \cite{Baxter}. The construction is similar to the one described above for the q-deformed Haldane-Shastry model
based on XXZ $R$-matrix. It uses recently found (commuting) set of
anisotropic generalizations of elliptic Ruijsenaars-Macdonald operators \cite{MZ}. In this case we have expression of the form (\ref{spMacd})
with\footnote{The minus in exponent in (\ref{spMacd0}) came from our first paper \cite{MZ}, where it appeared for consistency of
different evaluations.} $q=e^{2\pi\imath\eta}$:
\begin{equation}\label{spMacd0}
 \mD_k=\sum\limits_{\substack{|I|=k}}A_I\,{\bf R}_I \Big(\prod_{i\in I}e^{-\eta\p_{z_i}}\Big){\bf R}_I^{-1}\,.
\end{equation}
The coefficients $A_I$ in (\ref{Macd}) and (\ref{spMacd}) are as follows:
 \beq\label{s091}
 \begin{array}{c}
  \displaystyle{
A_I=\prod\limits_{\substack{i\in I \\ j\notin I}}\phi(z_j-z_i)\,,
  }
 \end{array}
\eq
 where $\phi(z)=\phi(\hbar,z)$ is the elliptic Kronecker function
 (\ref{a0962}), the parameter $\hbar$ is related to $t$ from (\ref{Macd}) as $t=e^{2\pi\imath\hbar}$ and $\phi(z)$ also depends on the elliptic modulus $\tau$ entering the definition of theta function
  (\ref{a0963}). The operators ${\bf R}_I$ are certain products
  of the elliptic Baxter-Belavin ${\rm GL}_M$ $R$-matrices (\ref{BB}). Precise expressions are give in the next Section. Expression (\ref{Macd})
  with the coefficients $A_I$ (\ref{s091}) provides the definition
  of the elliptic Ruijsenaars-Macdonald operators \cite{Ruij}:
\begin{equation}\label{Macd0}
 D_k=\sum\limits_{\substack{|I|=k}}A_I\prod_{i\in I}e^{-\eta\p_{z_i}}\,.
\end{equation}
  The spin many-body system related to operators (\ref{spMacd0}) in the elliptic case is presumably the relativistic model of interacting tops \cite{Z19} although
this statement needs further elucidation.

In \cite{Ruij} the definition of scalar operators (\ref{Macd0}) was given in a slightly different way\footnote{The definitions (\ref{Macd0}) and (\ref{Macd2}) are related through conjugation by a certain function.}:
\begin{equation}\label{Macd2}
 {D}_k'=\sum\limits_{\substack{|I|=k}}\sqrt{A_I}\Big(\prod_{i\in I}e^{-\eta\p_{z_i}}\Big)\sqrt{A_I'}\,,\quad
 A_I'=\prod\limits_{\substack{i\in I \\ j\notin I}}\phi(z_i-z_j)\,.
\end{equation}
Then the definition of spin operators is also modified:
\begin{equation}\label{spMacd2}
 {\mD}_k'=\sum\limits_{\substack{|I|=k}}\sqrt{A_I}\,{\bf R}_I \Big(\prod_{i\in I}e^{-\eta\p_{z_i}}\Big){\bf R}_I^{-1}\sqrt{A_I'}\,.
\end{equation}
Mutual commutativity of ${\mD}_k'$ is fulfilled as well (see \cite{MZ}).
We will call (\ref{Macd0}) the Ruijsenaars-Macdonald operators
in the Macdonald form, while (\ref{Macd2}) are the Ruijsenaars-Macdonald operators
in the Ruijsenaars form. And similarly for the spin operators
(\ref{spMacd0}) and (\ref{spMacd2}).

In this paper we show that the freezing trick applied to either (\ref{spMacd0})  or (\ref{spMacd2}) provides commutative set of
the long-range spin chain Hamiltonians if some set of elliptic functions identities is fulfilled.
Namely, following \cite{TH,Uglov} we consider expansion
$\mathcal{D}_k=\mathcal{D}_k^{[0]}+\eta\mathcal{D}_k^{[1]}+O(\eta^2)$
in $\eta$ of the operators $\mathcal{D}_k$ and introduce the operators
${\tilde H}_k$ through
$\mathcal{D}_k^{[1]}={\rm Id}\, D_k^{[1]}-{\tilde H}_k$. The expressions
${\tilde H}_k$ are free of differential operators. Being restricted to
the (equilibrium position) $z_k=x_k=k/N$ one obtains the set of Hamiltonians of the long-range spin chain. Their commutativity follows from a set of identities.
Then we prove these identities.
In the trigonometric limit the identities become quite simple, while in the elliptic case they are nontrivial. We also show that these identities provide the equilibrium position (in all Hamiltonian flows) in the underlying
classical elliptic Ruijsenaars-Schneider model.

{\bf Plan of the paper.} In Section \ref{sect2} we recall main result of \cite{MZ} including explicit expressions for the spin operators. In Section \ref{sect3} we study the freezing trick. Namely, we find identities, which lead to
commuting long-range spin chain Hamiltonians. These identities also guarantee that the set of points $x_j=j/N$ is an equilibrium position in the corresponding classical (and spinless) model. In the end of the Section \ref{sect3} we give detailed description of two first Hamiltonians and an example of $N=3$ sites case. In Section \ref{sect4} the set of elliptic function identities is proved. Section \ref{sect5} is devoted to
description of the Ruijsenaars formulation (\ref{spMacd2}) for the obtained results. We will show that the freezing trick works in this case as well and
provides the same set of the spin chain Hamiltonians.
In Section \ref{sect6} we explain which trigonometric limits are possible and show how the Uglov's q-deformed Haldane-Shastry model in ${\rm GL}_2$ case appear by reproducing its Hamiltonian in the form derived by J. Lamers in \cite{Lam}. Our construction is straightforwardly applicable for the model based on trigonometric $R$-matrix (\ref{y053}), (\ref{q203}). For example, when $c_7=0$ $R$-matrix (\ref{y053}) provides q-deformed version of the model (\ref{s09}) and its higher rank extensions. In order to obtain the
q-deformed Haldane-Shastry model one should use ${\rm U}_q({\widehat {\rm  gl}_2})$ $R$-matrix (\ref{q32}).
To include ${\rm U}_q({\widehat {\rm  gl}_M})$ $R$-matrix (\ref{q270}) into consideration we explain in the Appendix C that commutativity of the spin operators (\ref{spMacd2}) with this $R$-matrix holds true.
  In Section \ref{sect7} we discuss the non-relativistic limit and derive the first two nontrivial commuting Hamiltonians, which we obtained previously in \cite{SeZ} using $R$-matrix valued Lax pairs. A short summary
is given in the Conclusion. In the Appendix A some definitions and properties of elliptic functions and elliptic $R$-matrix are given. In the Appendix B
detailed expressions for the Hamiltonians in $N=4$ sites case are presented.

\section{Anisotropic spin Ruijsenaars-Macdonald operators}\label{sect2}
\setcounter{equation}{0}

\subsection{Elliptic $R$-matrix}
In this paper we deal with the elliptic ${\rm GL}_M$ Baxter-Belavin quantum $R$-matrix \cite{Baxter}.
It is given by
the expression (\ref{BB}). In $M=2$ case this is the Baxter's $R$-matrix for 8-vertex model:
 \beq\label{r724}
 \begin{array}{c}
  \displaystyle{
 R_{12}^\hbar(z)
 =\frac{1}{2}\Big(\vf_{00}\,\sigma_0\otimes\sigma_0
 +\vf_{01}\,\sigma_1\otimes\sigma_1
  +\vf_{11}\,\sigma_2\otimes\sigma_2
  +\vf_{10}\,\sigma_3\otimes\sigma_3\Big)\,,
  }
 \end{array}
\eq
 where $\sigma_a$, $a=0,1,2,3$ are the Pauli matrices ($\sigma_0=1_{2\times 2}$) and
 \beq\label{r826}
 \begin{array}{c}
  \displaystyle{
 \vf_{00}=\phi(z,\frac{\hbar}{2})\,,\quad
 \vf_{10}=\phi(z,\frac{1}{2}+\frac{\hbar}{2})\,,\quad
 \vf_{01}=e^{\pi\imath z}\phi(z,\frac{\tau}{2}+\frac{\hbar}{2})\,,\quad
 \vf_{11}=e^{\pi\imath z}\phi(z,\frac{1+\tau}{2}+\frac{\hbar}{2})\,,
 }
  \end{array}
 \eq
  so that it is $4\times 4$ matrix:
 \beq\label{r827}
 \begin{array}{c}
\displaystyle{
R_{12}^\hbar(z)=\frac{1}{2}
}
\left(
 \begin{array}{cccc}
 \vf_{00}+\vf_{10} & 0 & 0 & \vf_{01}-\vf_{11}
 \\
 0 & \vf_{00}-\vf_{10} & \vf_{01}+\vf_{11} & 0
  \\
 0 & \vf_{01}+\vf_{11} & \vf_{00}-\vf_{10} & 0
 \\
 \vf_{01}-\vf_{11} & 0 & 0 & \vf_{00}+\vf_{10}
 \end{array}
 \right)\,.
  \end{array}
 \eq
 Elliptic function notations (\ref{a0962}), (\ref{a08}) are used here. See Appendix for definitions.

 The $R$-matrix (\ref{BB}) satisfies the quantum Yang-Baxter equation
\beq\label{QYB}
\begin{array}{c}
\displaystyle{
    R^{\hbar}_{12}(u)  R^{\hbar}_{13}(u+v) R^{\hbar}_{23}(v) =
      R^{\hbar}_{23}(v) R^{\hbar}_{13}(u+v) R^{\hbar}_{12}(u)
      }
\end{array}\eq
and obeys the unitarity property
\beq\label{q03}\begin{array}{c}
    R^{\hbar}_{12}(z) R^\hbar_{21}(-z)= {\rm Id}\, \phi(\hbar,z)\phi(\hbar,-z)\stackrel{(\ref{a0964})}{=}
   {\rm Id} (\wp(\hbar)-\wp(z) ) \,.
\end{array}\eq
We also use the normalized $R$-matrix:
\beq\label{q04}
\begin{array}{c}
\displaystyle{
    {\bar R}^{\hbar}_{12}(z)  = \frac{1}{\phi(\hbar,z)} R^{\hbar}_{12}(z)\,.
    }
\end{array}\eq
Then (\ref{q03}) takes the form:
\beq\label{q05}
\begin{array}{c}
    {\bar R}^{\hbar}_{12}(z) {\bar R}^\hbar_{21}(-z)= {\rm Id}\,.
\end{array}\eq

\subsection{Commuting XYZ spin Ruijsenaars-Macdonald operators}

Following \cite{Uglov,LPS} we introduced in \cite{MZ} a set of anisotropic
spin Ruijsenaars-Macdonald operators:
\beq\label{q10}
\begin{array}{c}
  \displaystyle{
    {\mathcal D}_k=\sum\limits_{1\leq i_1<...<i_k\leq N}\left(\!\prod\limits^{N}_{\hbox{\tiny{$ \begin{array}{c}{ j=1 }\\{ j\!\neq\! i_1...i_{k-1} } \end{array}$}}}\!\phi(z_j-z_{i_1})\ \phi(z_j-z_{i_2})
    \ \cdots\
    \phi(z_j-z_{i_k})\right)\times
    }
    \\ \ \\
    \displaystyle{
    \times\left(
   \overleftarrow{\prod\limits_{j_1=1}^{i_1-1}} \bar{R}_{j_1 i_1}
   \overleftarrow{\prod\limits^{i_2-1}_{\hbox{\tiny{$ \begin{array}{c}{ j_2=1 }\\{ j_2\!\neq\! i_1 } \end{array}$}}}} \bar{R}_{j_2 i_2}
      \ \ldots\
 \overleftarrow{\prod\limits^{i_k-1}_{\hbox{\tiny{$ \begin{array}{c}{ j_k=1 }\\{ j_k\!\neq\! i_1...i_{k-1} } \end{array}$}}}} \bar{R}_{j_k i_k}
  \right)\times
   }
   \\ \ \\
     \displaystyle{
       \times p_{i_1}\cdot p_{i_2}\cdots p_{i_k}\times\left(
   \overrightarrow{\prod\limits^{i_k-1}_{\hbox{\tiny{$ \begin{array}{c}{ j_k\!=\!1 }\\{ j_k\!\neq\! i_{1}...i_{k-1}} \end{array}$}}}}\bar{R}_{i_k j_k}
   \overrightarrow{\prod\limits^{i_{k-1}-1}_{\hbox{\tiny{$ \begin{array}{c}{ j_{k-1}\!=\!1 }\\{ j_{k-1}\!\neq\! i_{1}...i_{k-2}} \end{array}$}}}}\bar{R}_{i_{k-1} j_{k-1}}
      \ \ldots\
  \overrightarrow{\prod\limits^{i_{1}-1}_ {j_1=1}} \bar{R}_{i_{1} j_{1}}\right),
 }
\end{array}\eq
where $k=1,...,N$, $\bar{R}_{ij}=\bar{R}_{ij}^\hbar(z_i-z_j)$ and $p_i$, $i=1,...,N$ are the shift operators
\beq\label{p_i}
(p_if)(z_1,z_2,\dots z_N)=\exp\left(-\eta \frac{\partial}{\partial z_i}\right)f(z_1,\dots,z_N)=f(z_1,\dots,z_i-\eta,\dots, z_N).
\eq
 The arrows in (\ref{q10}) mean the ordering in $R$-matrix products. For example,
$\overrightarrow{\prod\limits^{N}_ {j=1}} R_{ij}=R_{i,1}R_{i,2}...R_{i,N}$ and $\overleftarrow{\prod\limits^{N}_ {j=1}} R_{ji}=R_{N,i}R_{N-1,i}...R_{1,i}$. The notation ${\bf R}_I$ used
in (\ref{spMacd0}) and (\ref{spMacd2}) is as follows:
\beq\label{s200}
\begin{array}{c}
    \displaystyle{
   {\bf R}_I=\overleftarrow{\prod\limits_{j_1=1}^{i_1-1}} \bar{R}_{j_1 i_1}
   \overleftarrow{\prod\limits^{i_2-1}_{\hbox{\tiny{$ \begin{array}{c}{ j_2=1 }\\{ j_2\!\neq\! i_1 } \end{array}$}}}} \bar{R}_{j_2 i_2}
      \ \ldots\
 \overleftarrow{\prod\limits^{i_k-1}_{\hbox{\tiny{$ \begin{array}{c}{ j_k=1 }\\{ j_k\!\neq\! i_1...i_{k-1} } \end{array}$}}}} \bar{R}_{j_k i_k}\,,\quad I=\{i_1,...,i_k\}\,,\ |I|=k\,.
   }
\end{array}\eq
The $R$-matrix product in the brackets in the end of (the right hand side  of) (\ref{q10})
is equal to ${\bf R}_I^{-1}$ due to the unitarity (\ref{q05}), see details in \cite{MZ}.

The expressions  ${\mathcal D}_k$ (\ref{q10}) are  matrix-valued difference operators.
For example,
\beq\label{D1N}
\mathcal{D}_1=\sum_{i=1}^N \prod\limits_{\substack{j=1\\j\neq i}}^N\phi(z_j-z_i)\bar{R}_{i-1,i}\bar{R}_{i-2,i}\dots \bar{R}_{1,i}p_i \bar{R}_{i,1}\dots \bar{R}_{i,i-2}\bar{R}_{i,i-1}\,.
\eq
Any $R$-matrix $\bar{R}_{ij}$ acts non-trivially on the $i$-th and $j$-th tensor components of the Hilbert
space $\mH=(\mC^M)^{\otimes N}$. That is the quantum Yang-Baxter equation (\ref{QYB}) implies
\beq\label{QYB2}
\begin{array}{c}
\displaystyle{
    {\bar R}^{\hbar}_{ij}(z_i-z_j)  {\bar R}^{\hbar}_{ik}(z_i-z_k) {\bar R}^{\hbar}_{jk}(z_j-z_k) =
      {\bar R}^{\hbar}_{jk}(z_j-z_k) {\bar R}^{\hbar}_{ik}(z_i-z_k) {\bar R}^{\hbar}_{ij}(z_i-z_j)
      }
\end{array}\eq
for any distinct integers $1\leq i,j,k\leq N$ and
\beq\label{QYB3}
\begin{array}{c}
\displaystyle{
    [{\bar R}^{\hbar}_{ij}(u), {\bar R}^{\hbar'}_{kl}(v)]=0
      }
\end{array}\eq
for any distinct integers $1\leq i,j,k,l\leq N$. Thus, ${\mathcal D}_k$ (\ref{q10}) are ${\rm End}(\mH)$-valued
difference operators.

In \cite{MZ} the commutativity of operators (\ref{q10})
 \beq\label{a205}
  \begin{array}{c}
  \displaystyle{
   [{\mathcal D}_k,{\mathcal D}_l]=0\quad k,l=1,...,N
 }
 \end{array}
 \eq
 was proved to be equivalent to a set of identities, which were shown to be valid for the elliptic $R$-matrix (including some trigonometric and rational degenerations).
In the scalar case $M=1$, i.e. ${\bar R}^{\hbar}_{ij}={\rm Id}$. Then (\ref{q10}) turns
into the commuting set of the elliptic Ruijsenaars-Macdonald operators introduced in \cite{Ruij}:
\begin{equation}\label{Dscalar}
 D_k=\sum\limits_{\substack{|I|=k}}\prod\limits_{\substack{i\in I \\ j\notin I}}\phi(z_j-z_i)\prod_{i\in I}p_{i},\qquad k=1,\dots,N\,.
\end{equation}
The operators $D_k'$ and $\mD_k'$ in the Ruijsenaars form (both, scalar and spin) are obtained as given in (\ref{Macd2}) and (\ref{spMacd2}) with $A_I$ (\ref{s091}) and ${\bf R}_I$ (\ref{s200}). The operators $D_k'$ are those considered in \cite{Ruij}. These are the quantum Hamiltonians of the elliptic Ruijsenaars-Schneider model. We come back to discussion of $D_k'$ and $\mD_k'$ in Section 5.

\section{Freezing trick and commuting Hamiltonians for elliptic chain}\label{sect3}
\setcounter{equation}{0}

Here following ideas of \cite{Polych1,Uglov,LPS} we deduce the long-range spin chain Hamiltonians
and find conditions for their commutativity. These conditions are unified to a set of identities, which
are proved in the next Section. Here we describe the freezing trick.


\subsection{Classical spinless model: equilibrium position} Consider the classical analogues for elliptic Ruijsenaars-Macdonald operators (\ref{Dscalar}). These are the Hamiltonians of the classical elliptic Ruijsenaars-Schneider model:
\begin{equation}\label{s20}
 h_k=\sum\limits_{\substack{|I|=k}}\prod\limits_{\substack{i\in I \\ j\notin I}}\phi(z_j-z_i)\prod_{i\in I}e^{-v_i/c}\qquad k=1,\dots,N\,,
\end{equation}
where $c$ is a constant (light speed\footnote{The Hamiltonian comes as a classical limit $\kappa\rightarrow 0$ of (\ref{Dscalar}) with $\eta=\kappa/c$.}), and $v_1,...,v_N$ are momenta canonically conjugated to the positions of particles $z_1,...,z_N$:
\begin{equation}\label{s21}
 \{v_i,z_j\}=\delta_{ij}\,,\quad \{v_i,v_j\}=\{z_i,z_j\}=0\,.
\end{equation}
The classical integrability means $\{h_k,h_l\}=0$ for any $k,l=1,...,N$.
Each Hamiltonian $h_k$ provides its dynamics (Hamiltonian flow) through
the Hamiltonian equations
\begin{equation}\label{s22}
\displaystyle{
  \frac{dz_j}{dt_k}=\{h_k,z_j\}=\frac{\p h_k}{\p v_j}\,,\qquad \frac{dv_j}{dt_k}=\{h_k,v_j\}=-\frac{\p h_k}{\p z_j}\,.
  }
\end{equation}
The velocities are as follows:
\begin{equation}\label{s221}
\displaystyle{
  \frac{dz_m}{dt_k}=-\frac{1}{c}\sum\limits_{\substack{|I|=k\\ m\in I}}\prod\limits_{\substack{i\in I\\j\notin I}}\phi(z_j-z_i)\prod_{i\in I}e^{-v_i/c}\,.
  }
\end{equation}

Consider the following special point in the phase space:
\begin{equation}\label{s23}
\displaystyle{
  {\rm eq}:\ z_k=x_k:=\frac{k}{N}\,,\quad v_k=0\,,\quad k=1,...,N\,,
  }
\end{equation}
where ''eq'' comes from either equidistant or equilibrium.
Below we will see that restriction of equations of motion (\ref{s22}) to  (\ref{s23})
provides a kind of equilibrium position in the following sense. Denote
the set of classical velocities restricted to (\ref{s23}) as $u^{\{k\}}_j$:
\begin{equation}\label{s24}
\displaystyle{
  u^{\{k\}}_m=c\frac{dz_m}{dt_k}\Bigg|_{\rm eq}=-\sum\limits_{\substack{|I|=k\\ m\in I}}\prod\limits_{\substack{i\in I\\j\notin I}}\phi(x_j-x_i)\,.
  }
\end{equation}
Similarly, denote
\begin{equation}\label{s25}
\displaystyle{
  w^{\{k\}}_m=-\frac{dv_m}{dt_k}\Big|_{\rm eq}=\frac{\p h_k}{\p z_m}\Big|_{\rm eq}\,.
  }
\end{equation}
We will show that for any $k,l,m=1,...,N$ the velocities in $k$-th flow being restricted to (\ref{s23})
are equal to each other (see also trigonometric analogue of this statement (\ref{q39}))
\begin{equation}\label{s26}
\displaystyle{
  u^{\{k\}}_m=u^{\{k\}}_l
  }
\end{equation}
and
\begin{equation}\label{s27}
\displaystyle{
  w^{\{k\}}_m=0\,.
  }
\end{equation}
Let us mention that the accelerations ${\ddot z}_m$ also vanish on (\ref{s23}). Indeed, by differentiating  (\ref{s221}) with respect to time variable and then using (\ref{s26})-(\ref{s27}) one immediately gets
\begin{equation}\label{s222}
\displaystyle{
  \frac{d^2z_m}{dt^2_k}\Bigg|_{\rm eq}=0\,,\quad k,m=1,...,N
  }
\end{equation}
if (\ref{s26})-(\ref{s27}) hold true.
Although velocities do not vanish, they are equal to each other in any flow. So that
$z_i(t_k)=u^{\{k\}}_i t_k+x_k$ and $z_i(t_k)-z_j(t_k)=x_i-x_j$. Therefore, we may consider such configuration as an equilibrium position in the linearly moving frame.

\subsection{Quantum spinless model} The elliptic Ruijsenaars-Macdonald operators (\ref{Dscalar}) mutually commute:
\begin{equation}\label{s28}
\displaystyle{
[D_k,D_l]=0\,,\quad k,l=1,...,N\,.
  }
\end{equation}
Consider expansion of $D_k$ in variable $\eta$ (near $\eta=0$):
\begin{equation}\label{s29}
\displaystyle{
  D_k=D_k^{[0]}+\eta D_k^{[1]}+\eta^2 D_k^{[2]}+O(\eta^3)\,,\quad k=1,...,N\,.
  }
\end{equation}
Using also (\ref{p_i}), i.e.
\beq\label{s291}
p_i=1-\eta\frac{\partial}{\partial z_i} +\mathcal{O}(\eta^2)
\eq
we find that $D_k^{[0]}$ are the following functions
\begin{equation}\label{s30}
\displaystyle{
  D_k^{[0]}=\sum_{|I|=k}\prod\limits_{\substack{i\in I\\j\notin I}}\phi(z_j-z_i)\,,\quad k=1,...,N
  }
\end{equation}
and
\begin{equation}\label{s31}
\displaystyle{
  D_k^{[1]}=-\sum\limits_{m=1}^N\Big(\sum\limits_{\substack{|I|=k\\ m\in I}}\prod\limits_{\substack{i\in I\\j\notin I}}\phi(z_j-z_i)\Big)\p_{z_m}\,.
  }
\end{equation}
Similarly, $D_k^{[l]}$ are some $l$-th order differential operators. From the commutativity (\ref{s28}) we conclude
\begin{equation}\label{s32}
\displaystyle{
  [D_i^{[0]},D_j^{[1]}]+[D_i^{[1]},D_j^{[0]}]=0\,,\quad i,j=1,...,N
  }
\end{equation}
and
\begin{equation}\label{s33}
\displaystyle{
  [D_i^{[0]},D_j^{[2]}]+[D_i^{[1]},D_j^{[1]}]+[D_i^{[2]},D_j^{[0]}]=0\,,\quad i,j=1,...,N\,.
  }
\end{equation}

\subsection{Quantum spin model}
\paragraph{Expansion of spin operators.} Consider now the spin operators (\ref{q10}) similarly to the previous subsection.
Namely, consider expansion of $\mD_k$ in variable $\eta$ (near $\eta=0$):
\begin{equation}\label{s34}
\displaystyle{
  \mD_k=\mD_k^{[0]}+\eta \mD_k^{[1]}+\eta^2 \mD_k^{[2]}+O(\eta^3)\,,\quad k=1,...,N\,.
  }
\end{equation}
Since $\mD_k^{[0]}=\mD_k|_{\eta=0}$ and due to the unitarity (\ref{q05}) we have
\beq\label{s35}
\mathcal{D}_k^{[0]}={\rm Id}\sum_{|I|=k}\prod\limits_{\substack{i\in I\\j\notin I}}\phi(z_j-z_i)={\rm Id}\,D_k^{[0]}\,,
\eq
where ${\rm Id}$ is the identity matrix in ${\rm End}(\mH)$.
For the set of $\mD_k^{[1]}$ one gets
\beq\label{s36}
\begin{array}{c}
\displaystyle{
-\mathcal{D}_1^{[1]}={\rm Id}\sum_{i=1}^N \prod\limits_{\substack{j=1\\j\neq i}}^N\phi(z_j-z_i)\frac{\partial}{\partial z_i}+
}
\\ \ \\
\displaystyle{
+\sum_{i=1}^N \prod\limits_{\substack{j=1\\j\neq i}}^N\phi(z_j-z_i)\sum_{k=1}^{i-1}\bar{R}_{i-1,i}\dots \bar{R}_{k+1,i}\bar{R}_{k,i}\left(\frac{\partial}{\partial z_i}\bar{R}_{i,k}\right)\bar{R}_{i,k+1}\dots \bar{R}_{i,i-1}\,.
}
\end{array}
\eq
For $k=2,...,N-1$:
\beq\label{s37}
-\mathcal{D}_k^{[1]}={\rm Id}\sum_{l=1}^N \left(\sum\limits_{\substack{|I|=k\\ l\in I}}\prod\limits_{\substack{i\in I\\j\notin I}}\phi(z_j-z_i)\right) \frac{\partial}{\partial z_l}+ \text{(terms with $R$-matrices and their derivatives)}\,.
\eq
In particular, for $k=2$ from (\ref{q10}) we get:
\beq\label{s37b}
\begin{array}{c}
\displaystyle{
-\mathcal{D}_2^{[1]}={\rm Id}\sum_{m=1}^N \left(\sum\limits_{\substack{l=1\\l\neq m}}^N\prod\limits_{\substack{j=1\\j\neq m\\j\neq l}}^N\phi(z_j-z_m)\phi(z_j-z_l)\right)\frac{\partial}{\partial z_m}
+\sum\limits_{\substack{m,l=1\\m<l}}^N \prod\limits_{\substack{j=1\\j\neq m\\j\neq l}}^N\phi(z_j-z_m)\phi(z_j-z_l)\times
}
\\ \ \\
\displaystyle{
\times\left(\sum_{i=1}^{m-1}\bar{R}_{m-1,m}\dots \bar{R}_{i+1,m}\bar{R}_{i,m}\left(\frac{\partial}{\partial z_m}\bar{R}_{m,i}\right)\bar{R}_{m,i+1}\dots \bar{R}_{m,m-1}+\right.
}
\end{array}
\eq
$$
\begin{array}{c}
\displaystyle{
+\sum_{i=1}^{m-1}\bar{R}_{m-1,m}\dots \bar{R}_{1,m}\bar{R}_{l-1,l}\dots \bar{R}_{m+1,l}\bar{R}_{m-1,l}\dots\bar{R}_{i+1,l}\bar{R}_{i,l}\times}
\\ \ \\
{\displaystyle \times\left(\frac{\partial}{\partial z_l}\bar{R}_{l,i}\right)\bar{R}_{l,i+1}\dots\bar{R}_{l,m-1}\bar{R}_{l,m+1}\dots \bar{R}_{l,l-1}\bar{R}_{m,1}\dots \bar{R}_{m,m-1}+
}
\\ \ \\
\displaystyle{
\left.+\sum_{i=m+1}^{l-1}\bar{R}_{l-1,l}\dots\bar{R}_{i+1,l}\bar{R}_{i,l}\left(\frac{\partial}{\partial z_l}\bar{R}_{l,i}\right)\bar{R}_{l,i+1}\dots \bar{R}_{l,l-1}\right)\,.
}
\end{array}
$$
Finally, for $k=N$:
\beq\label{s38}
-\mathcal{D}_N^{[1]}={\rm Id}\sum_{i=1}^N \frac{\partial}{\partial z_i}\,.
\eq
It is easy to see from (\ref{s36})-(\ref{s38}) and (\ref{s31}) that
\beq\label{s39}
\displaystyle{
\mathcal{D}_k^{[1]}={\rm Id}\, D_k^{[1]}-{\tilde H}_k\,,\quad k=1,...,N\,,
}
\eq
where ${\tilde H}_k\in{\rm End}(\mH)$ are some matrix-valued functions, which contain $R$-matrix derivatives but do not contain differential operators.
By performing similar calculations for $\mD^{[2]}_k$ we come to
\beq\label{s40}
\displaystyle{
\mathcal{D}_k^{[2]}={\rm Id}\, D_k^{[2]}+\sum\limits_{l=1}^N A_{k,l}\p_{z_l} +B_{k}\,,
}
\eq
where $A_{k,l}$ and $B_k$ are again some matrix-valued functions free of differential operators. Expressions
$A_{k,l}$ contain one derivative of $R$-matrix, and expressions $B_k$ contain two derivatives of (one or two) $R$-matrices.

\paragraph{Commuting Hamiltonians.} Similarly to (\ref{s32})-(\ref{s33}) due to commutativity of $\mD_k$ (\ref{a205}) we have
\begin{equation}\label{s41}
\displaystyle{
  [\mD_i^{[0]},\mD_j^{[1]}]+[\mD_i^{[1]},\mD_j^{[0]}]=0\,,\quad i,j=1,...,N
  }
\end{equation}
and
\begin{equation}\label{s42}
\displaystyle{
  [\mD_i^{[0]},\mD_j^{[2]}]+[\mD_i^{[1]},\mD_j^{[1]}]+[\mD_i^{[2]},\mD_j^{[0]}]=0\,,\quad i,j=1,...,N\,.
  }
\end{equation}
Plugging expressions for $\mD_i^{[0]}$, $\mD_i^{[1]}$ and $\mD_i^{[2]}$ from (\ref{s35}), (\ref{s39}) and (\ref{s40}) into  (\ref{s41}) and (\ref{s42}) we see
that (\ref{s41}) is fulfilled due to $\mD_i^{[0]}$ is proportional to ${\rm Id}$ and due to (\ref{s32}).
Consider (\ref{s42}). Due to  (\ref{s33}) we get
\begin{equation}\label{s43}
\displaystyle{
 [{\tilde H}_i,{\tilde H}_j]-[D_i^{[1]},{\tilde H}_j]+
  [D_j^{[1]},{\tilde H}_i]+
  \sum\limits_{l=1}^N A_{i,l}(\p_{z_l}D_j^{[0]})-\sum\limits_{l=1}^N A_{j,l}(\p_{z_l}D_i^{[0]})=0\,.
  }
\end{equation}
Let us now restrict the latter equality to the point (\ref{s23}) and denote
\begin{equation}\label{s431}
\displaystyle{
 H_i={\tilde H}_i\Big|_{\rm eq}\,.
  }
\end{equation}
Notice that using notation (\ref{s24})
\begin{equation}\label{s44}
\displaystyle{
 -D_i^{[1]}\Big|_{\rm eq}=\sum\limits_{m=1}^N u^{\{i\}}_{m}\p_{z_m}\,.
  }
\end{equation}
Therefore, if (\ref{s26}) holds true, then all $D_i^{[1]}\Big|_{\rm eq}$ are proportional to $-D_N^{[1]}=\p_{z_1}+...+\p_{z_N}$. Thus, $[D_i^{[1]},{\tilde H}_j]$
  and $[D_j^{[1]},{\tilde H}_i]$ vanish on (\ref{s23}) if (\ref{s26}) holds.
  Also, the sums with coefficients $A$ in (\ref{s43}) vanish on (\ref{s23}) if (\ref{s27}) hold.
  In this way we come to commutativity
\begin{equation}\label{s45}
\displaystyle{
 [H_i,H_j]=0\,,\quad i,j=1,...,N-1
  }
\end{equation}
on (\ref{s26})-(\ref{s27}). $H_N=0$ is trivial.

To summarize we proved the following
\begin{predl}
Consider the expansions in $\eta$ of the scalar (\ref{s29}) and spin (\ref{s34}) operators. Then the expressions (\ref{s39}) ${\tilde H}_k={\rm Id}\, D_k^{[1]}-\mathcal{D}_k^{[1]}$ are matrix valued functions free of differential operators.  Being restricted to the point (\ref{s23}) one gets
the set (\ref{s431}) of commuting operators (\ref{s45}) if the identities (\ref{s26})-(\ref{s27}) hold true. The latter identities define the equilibrium position in the classical model (\ref{s20}).
\end{predl}

The above results mean that we proved quantum integrability for ${\rm GL}_2$ spin chain. However,  for ${\rm GL}_M$ case with $M>2$ the presented set of $N-1$ commuting Hamiltonians is not enough. There should be more commuting Hamiltonians.
Their existence follows from studies of the classical model of relativistic tops \cite{Z19}, which has Lax representation with spectral parameter. Presumably, it is the classical version of the model described by spin operators (\ref{q10}). Let also remark that we deal with fundamental representations in this paper. More complicated representations need to be studied as well.  We hope to clarify these questions in future papers.

\subsection{Examples}

\subsubsection*{The first and the second Hamiltonians}
In this subsection we assume the following short notations:
\beq\label{s4601}
\begin{array}{c}
\displaystyle{
{\bar F}^\hbar_{ij}(z)=\frac{\partial}{\partial z}\bar{R}^\hbar_{ij}(z)
}
\end{array}
\eq
and
\beq\label{s4602}
\begin{array}{c}
\displaystyle{
{\bar R}_{ij}={\bar R}^\hbar_{ij}(x_i-x_j)\,,\qquad
{\bar F}_{ij}={\bar F}^\hbar_{ij}(x_i-x_j)\,.
}
\end{array}
\eq
Let us write down expression for the first Hamiltonian (\ref{s431}):
\beq\label{s46}
\begin{array}{c}
\displaystyle{
H_1=\sum_{i=1}^N \prod\limits_{\substack{j=1\\j\neq i}}^N\phi(x_j-x_i)\sum_{k=1}^{i-1}\bar{R}_{i-1,i}\dots \bar{R}_{k+1,i}\bar{R}_{k,i} \bar{F}_{i,k}\bar{R}_{i,k+1}\dots \bar{R}_{i,i-1}\,.
}
\end{array}
\eq
Notice that the coefficients $\prod\limits_{\substack{j:j\neq i}}^N\phi(x_j-x_i)$ are equal to $-u_i^{\{1\}}$ (\ref{s24}) and equal to each other for any $i$ (\ref{s26}). Therefore, we may divide $H_1$
by $-u_i^{\{1\}}$ and obtain:
\beq\label{s461}
\begin{array}{c}
\displaystyle{
{\bf H}_1=-\frac{1}{u_i^{\{1\}}}\,H_1 =\sum\limits_{k<i}^N \bar{R}_{i-1,i}\dots \bar{R}_{k+1,i}\bar{R}_{k,i}\bar{F}_{i,k}\bar{R}_{i,k+1}\dots \bar{R}_{i,i-1}\,.
}
\end{array}
\eq
The second Hamiltonian is as follows:
\beq\notag
\begin{array}{c}
\displaystyle{
H_2=\sum\limits_{\substack{m,l=1\\m<l}}^N \prod\limits_{\substack{j=1\\j\neq m,l}}^N\phi(x_j-x_m)\phi(x_j-x_l)\times
}
\\
\displaystyle{
\times\left(\sum_{i=1}^{m-1}\bar{R}_{m-1,m}\dots \bar{R}_{i+1,m}\bar{R}_{i,m}\bar{F}_{m,i}\bar{R}_{m,i+1}\dots \bar{R}_{m,m-1}+\right.
}
\end{array}
\eq
\beq\label{s47}
\begin{array}{c}
\displaystyle{
+\sum_{i=1}^{m-1}\bar{R}_{m-1,m}\dots \bar{R}_{1,m}\bar{R}_{l-1,l}\dots \bar{R}_{m+1,l}\bar{R}_{m-1,l}\dots\bar{R}_{i+1,l}\bar{R}_{i,l}\times}
\\ \ \\
{\displaystyle\times\bar{F}_{l,i}\bar{R}_{l,i+1}\dots\bar{R}_{l,m-1}\bar{R}_{l,m+1}\dots \bar{R}_{l,l-1}\bar{R}_{m,1}\dots \bar{R}_{m,m-1}+
}
\end{array}
\eq
\beq\notag
\begin{array}{c}
\displaystyle{
\left.+\sum_{i=m+1}^{l-1}\bar{R}_{l-1,l}\dots\bar{R}_{i+1,l}\bar{R}_{i,l}\bar{F}_{l,i}\bar{R}_{l,i+1}\dots \bar{R}_{l,l-1}\right)\,.
}
\end{array}
\eq
Similarly to the first Hamiltonian we mention that for any
$1\leq m,l\leq N$
\beq\label{s471}
\begin{array}{c}
\displaystyle{
\prod\limits_{\substack{j=1\\j\neq m,l}}^N\phi(x_j-x_m)\phi(x_j-x_l)=
\frac{u_m^{\{1\}}}{\phi(x_l-x_m)}\frac{u_l^{\{1\}}}{\phi(x_m-x_l)}
\stackrel{(\ref{a0964})}{=}\frac{u_m^{\{1\}}u_l^{\{1\}}}{\wp(\hbar)-\wp(x_m-x_l)}\,.
}
\end{array}
\eq
Using again (\ref{s26}) and dividing $H_2$ by $u_m^{\{1\}}u_l^{\{1\}}$
we may redefine the second Hamiltonian as follows:
\beq\notag
\begin{array}{c}
\displaystyle{
{\bf H}_2=\frac{1}{u_m^{\{1\}}u_l^{\{1\}}}\,H_2=
}
\\ \ \\
\displaystyle{
=\sum\limits_{\substack{i,m,l=1\\i<m<l}}^N
\frac{1}{\wp(\hbar)-\wp(x_m-x_l)}
\left(\bar{R}_{m-1,m}\dots \bar{R}_{i+1,m}\bar{R}_{i,m}\bar{F}_{m,i}\bar{R}_{m,i+1}\dots \bar{R}_{m,m-1}+\right.
}
\end{array}
\eq
\beq\label{s472}
\begin{array}{c}
\displaystyle{
+\bar{R}_{m-1,m}\dots \bar{R}_{1,m}\bar{R}_{l-1,l}\dots \bar{R}_{m+1,l}\bar{R}_{m-1,l}\dots\bar{R}_{i+1,l}\bar{R}_{i,l}\times}
\\ \ \\
\displaystyle{\left.\times\bar{F}_{l,i}\bar{R}_{l,i+1}\dots\bar{R}_{l,m-1}\bar{R}_{l,m+1}\dots \bar{R}_{l,l-1}\bar{R}_{m,1}\dots \bar{R}_{m,m-1}\right)+
}
\end{array}
\eq
\beq\notag
\begin{array}{c}
\displaystyle{
+
\sum\limits_{\substack{i,m,l=1\\i<m<l}}^N \frac{1}{\wp(\hbar)-\wp(x_i-x_l)}\bar{R}_{l-1,l}\dots\bar{R}_{m+1,l}\bar{R}_{m,l}\bar{F}_{l,m}\bar{R}_{l,m+1}\dots \bar{R}_{l,l-1}\,.
}
\end{array}
\eq

\subsubsection*{Hamiltonians for $N=3$}
%
The Hamiltonians (\ref{s46}), (\ref{s47}):
\beq\label{ex01}
 \begin{array}{lll}
 \displaystyle H_1&=&\phi(x_1 - x_2) \phi(x_3 - x_2) \bar{R}^\hbar_{12}(x_1 - x_2)\bar{F}^\hbar_{21}(x_2-x_1)+
      \\ \ \\
    & &  + \phi(x_1 - x_3) \phi(x_2 - x_3)\bar{R}^\hbar_{23}(x_2 - x_3) \bar{R}^\hbar_{13}(x_1 - x_3)\bar{F}^\hbar_{31}(x_3-x_1) \bar{R}^\hbar_{32}(x_3 - x_2)+
       \\ \ \\
    & &   + \phi(x_1 - x_3) \phi(x_2 - x_3
      )  \bar{R}^\hbar_{23}(x_2 - x_3)\bar{F}^\hbar_{32}(x_3-x_2)\,,
  \end{array}
 \eq
  \beq\label{ex02}
 \begin{array}{lll}
H_2&=& \phi(x_2 - x_1) \phi(x_2 - x_3) \bar{R}^\hbar_{23}(x_2 - x_3)\bar{F}^\hbar_{32}(x_3-x_2)+
         \\ \ \\
    & & +\phi(x_1 - x_2) \phi(x_1 - x_3    ) \bar{R}^\hbar_{12}(x_1 - x_2) \bar{R}^\hbar_{13}(x_1 - x_3)\bar{F}^\hbar_{31}(x_3-x_1)
\bar{R}^\hbar_{21}(x_2 - x_1)+
       \\ \ \\
    & & + \phi(x_1 - x_2) \phi(x_1 - x_3       )  \bar{R}^\hbar_{12}(x_1 - x_2)\bar{F}^\hbar_{21}(x_2-x_1)\,.
     \end{array}
 \eq
  A modified version of these Hamiltonians (\ref{s461}), (\ref{s472}):
    \beq\label{ex01a}
 \begin{array}{lll}
 \displaystyle {\bf H}_1 &=&\bar{R}^\hbar_{12}(x_1 - x_2)\bar{F}^\hbar_{21}(x_2-x_1)+ \bar{R}^\hbar_{23}(x_2 - x_3)\bar{F}^\hbar_{32}(x_3-x_2)+
      \\ \ \\
&  & + \bar{R}^\hbar_{23}(x_2 - x_3) \bar{R}^\hbar_{13}(x_1 - x_3)\bar{F}^\hbar_{31}(x_3-x_1) \bar{R}^\hbar_{32}(x_3 - x_2)\,,
\qquad\qquad\qquad\quad\
  \end{array}
 \eq
 \beq\label{ex02a}
 \begin{array}{lll}
{\bf H}_2&=& \displaystyle {\frac{1}{\wp(\hbar)-\wp(\frac{1}{3})} \left(\bar{R}^\hbar_{23}(x_2 - x_3)\bar{F}^\hbar_{32}(x_3-x_2)+\bar{R}^\hbar_{12}(x_1 - x_2)\bar{F}^\hbar_{21}(x_2-x_1)+\right.}
         \\ \ \\
 & & \displaystyle{  \left.+\bar{R}^\hbar_{12}(x_1 - x_2) \bar{R}^\hbar_{13}(x_1 - x_3)\bar{F}^\hbar_{31}(x_3-x_1)
\bar{R}^\hbar_{21}(x_2 - x_1)\right)\,.}
     \end{array}
 \eq
The Hamiltonians for $N=4$ are given in the Appendix.

\section{Proof of elliptic function identities}\label{sect4}
\setcounter{equation}{0}
In this paragraph we prove (\ref{s26}) and (\ref{s27}).
First, we formulate and prove the following statement equivalent to (\ref{s26}):
\begin{lemma}\label{lemma1}
For $x_j=\frac{j}{N}$ the following relation holds:
 \beq \label{t604}
\sum\limits_{\substack{|I|=k\\ l\in I}}\prod\limits_{\substack{i\in I\\j\notin I}}\phi(x_j-x_i)=\sum\limits_{\substack{|I'|=k\\ m\in I'}}\prod\limits_{\substack{i\in I'\\j\notin I'}}\phi(x_j-x_i) \qquad \text{for  } l,m=1\dots N.
\eq
\end{lemma}
\begin{proof}
 Consider the map \footnotemark{}
 \beq\label{t605}
 \alpha:i\to i+(m-l) \mod N,
 \eq
 which sets a bijection between $k$-element subsets $I$ of the set $\{1,\dots,N\}$  containing the index $l$ on one hand, and  $k$-element subsets  $I'$ containing the index $m$ on the other hand.
  \footnotetext{Unlike the usual notation $b=a \mod N$ we mean that $b$ is a natural number from $1$ to $N$ such that $b=a\pmod N$. }
 Let us show that the summand with $I=\{i_1,\dots i_k\}$ in the left hand side  of (\ref{t604}) and the summand with $I'=\{\alpha(i_1),\dots ,\alpha(i_k)\}$ (here the elements are not ordered)  in the right hand side  of (\ref{t604}) are equal.
  Since the complement sets are connected as $(I')^c=\alpha(I^c)$, the difference between the coordinates
 $(x_j-x_i)$ and $(x_{\alpha(j)}-x_{\alpha(i)})$ differs by an integer number.
Due to periodic properties (\ref{a096}) we have:
\beq
\phi(x_j-x_i)=\phi(x_{\alpha(j)}-x_{\alpha(i)}).
\eq
Since (\ref{t605}) is a bijection one obtains (\ref{t604}). This finishes the proof.
\end{proof}

Before proceeding to (\ref{s27}) we prove another set of relations. Denote by
\beq\label{t10}
g(x)=E_1(\hbar+x)-E_1(x) \,,
\eq
and
\beq\label{t11}
f(x)=g(x)-g(-x)=E_1(\hbar+x)+E_1(\hbar-x)-2E_1(x) \,,
\eq
where $E_1(x)$ is given by (\ref{serE}).
Due to periodic behaviour of the first Eisenstein function $E_1(x)=E_1(x+1)$ and its skew-symmetry the same properties  are valid
for $f(x)$:
\beq\label{t12}
f(x)=f(x+1)\,,
\eq
\beq\label{t13}
f(x)=-f(-x)\,.
\eq
Further we use short notation  $f_{ij}$ for  $f(x_i-x_j)$, where $x_k$ are special points given in (\ref{s23}).
\begin{lemma}
The following identities hold
\beq\label{t14}
\sum_{l\neq m}f_{lm}=0\, ,
\eq
\beq\label{t15}
\sum\limits_{\substack{|I|=k\\m\in I}}\left(\prod\limits_{\substack{i\in I \\ j\notin I}}\phi(x_j-x_i) \sum_{\substack{l\in I\\l\neq m}} f_{lm}\right)=0\,.
 \eq
\end{lemma}
\begin{proof}
Let $l'=2m-l \mod N$, we show that
\beq\label{t16}
f_{l,m}=f_{m,l'}.
\eq
Indeed,
$$
l-m=m-l'\pmod N
$$
thus $(x_l-x_m)=\frac{l-m}{N}$ and $(x_m-x_{l'})=\frac{m-l'}{N}$ differs by integer number and due to periodic property (\ref{t12}) identity (\ref{t16}) holds
due to (\ref{t12}) and (\ref{t13}). The doubled left hand side  of (\ref{t14})
\beq\label{t17}
2\sum_{l\neq m}f_{lm}=\sum_{l\neq m}(f_{lm}+f_{l'm})=\sum_{l\neq m}(f_{lm}-f_{m,l'})=0,
\eq
vanishes due to skew-symmetric property (\ref{t13}), which proves (\ref{t14}).

To prove (\ref{t15}) using the same arguments as for (\ref{t14}) we need to show that in (\ref{t15}) the coefficient at $f_{lm}$ is equal to  the coefficient at $f_{m,l'}$ :
\beq\label{t18}
\sum\limits_{\substack{|I|=k\\m,l\in I}}\prod\limits_{\substack{i\in I \\ j\notin I}}\phi(x_j-x_i)=
\sum\limits_{\substack{|I'|=k\\m,l'\in I'}}\prod\limits_{\substack{i\in I' \\ j\notin I'}}\phi(x_j-x_i)\,.
\eq
The map
 $$
 \alpha:i\to i+(m-l) \mod N
 $$
 set a bijection between $k$-element subsets $I$  and $k$-element subsets $I'$ (and its complement sets respectively) such that $m,l\in I$  and  $m,l'\in I'$. Indeed, $\alpha(m)=l'$ and $\alpha(l)=m$. Since $\alpha(i)-\alpha(j)=i-j\pmod N$ and due to periodic properties (\ref{a096}):
 \beq\label{t19}
 \phi(x_{\alpha(i)}- x_{\alpha(j)})=\phi(x_i-x_j)\,.
 \eq
 Thus (\ref{t18}) holds which finishes the proof.
\end{proof}

 Using (\ref{a0965}) and the definition (\ref{t10}) we have
 \beq\label{t20}
 \frac{\p}{\p z}\phi(z)=\phi(z)g(z)\,,
 \eq
 where function $g(z)$ is from (\ref{t10}).
 In short notations :
 \beq\label{t21}
 \frac{\p}{\p z_i}\phi(z_j-z_i)\Big|_{\rm eq}=-\phi_{ji}g_{ji} \qquad \frac{\p}{\p z_j}\phi(z_j-z_i)\Big|_{\rm eq}=\phi_{ji}g_{ji}\,.
 \eq
 \begin{lemma}
 The expressions $w^{\{k\}}_m$ (\ref{s25}) are equal to zero (\ref{s27}).
 \end{lemma}
 \begin{proof}
Due to (\ref{t21}) $ w^{\{k\}}_m$ from (\ref{s25}) can be rewritten as
\beq\label{t61}
  w^{\{k\}}_m=\frac{\p h_k}{\p z_m}\Bigg|_{\rm eq}=-\sum\limits_{\substack{|I|=k\\m\in I}}\left(\prod\limits_{\substack{i\in I \\ j\notin I}}\phi_{ji} \sum_{l\notin I} g_{lm}\right)+\sum\limits_{\substack{|I|=k\\m\notin I}}\left(\prod\limits_{\substack{i\in I \\ j\notin I}}\phi_{ji} \sum_{\substack{l\in I\\l\neq m}} g_{ml}\right)\,.
\eq
Now let us show that the right hand side  of (\ref{t61}) equals zero, which proves (\ref{s27}).
We add and subtract to  $w^{\{k\}}_m$ the item $\displaystyle-\sum\limits_{\substack{|I|=k\\m\in I}}\left(\prod\limits_{\substack{i\in I \\ j\notin I}}\phi_{ji} \sum_{\substack{l\in I\\l\neq m}} g_{lm}\right)$  and combine two sums in $u^{\{k\}}_m$:
\beq\label{t62}
   w^{\{k\}}_m=u^{\{k\}}_m\sum_{l\neq m}g_{lm}+\sum\limits_{\substack{|I|=k\\m\in I}}\left(\prod\limits_{\substack{i\in I \\ j\notin I}}\phi_{ji} \sum_{\substack{l\in I\\l\neq m}} g_{lm}\right)+\sum\limits_{\substack{|I|=k\\m\notin I}}\left(\prod\limits_{\substack{i\in I \\ j\notin I}}\phi_{ji} \sum_{\substack{l\in I\\l\neq m}} g_{ml}\right)=
\eq
\beq\label{t63}
 =u^{\{k\}}_m\sum_{l\neq m}g_{lm}+\sum\limits_{\substack{|I|=k\\m\in I}}\left(\prod\limits_{\substack{i\in I \\ j\notin I}}\phi_{ji} \sum_{\substack{l\in I\\l\neq m}} f_{lm}\right)+\sum\limits_{\substack{|I|=k\\m\notin I}}\left(\prod\limits_{\substack{i\in I \\ j\notin I}}\phi_{ji} \sum_{\substack{l\in I\\l\neq m}} g_{ml}\right)+\sum\limits_{\substack{|I|=k\\m\in I}}\left(\prod\limits_{\substack{i\in I \\ j\notin I}}\phi_{ji}\sum_{\substack{l\in I\\l\neq m}} g_{ml}\right)=
\eq

\beq\label{t64}
 =u^{\{k\}}_m\sum_{l\neq m}g_{lm}+\sum\limits_{\substack{|I|=k\\m\in I}}\left(\prod\limits_{\substack{i\in I \\ j\notin I}}\phi_{ji} \sum_{\substack{l\in I\\l\neq m}} f_{lm}\right)+\sum\limits_{\substack{|I|=k}}\left(\prod\limits_{\substack{i\in I \\ j\notin I}}\phi_{ji} \sum_{\substack{l\in I\\l\neq m}} g_{ml}\right)\,.
\eq
In the first equality between (\ref{t62}) and (\ref{t63}) we used $g_{lm}=f_{lm}+g_{ml}$, and in the second equality between (\ref{t63}) and (\ref{t64}) the two last sums over subsets in $k$ elements, where $m\in I$ and $m\notin I$ were combined into one sum over all subsets in $k$ elements. Notice that the second sum in (\ref{t64})
vanishes due to (\ref{t15}). Consider the last sum in (\ref{t64}).
Since
$$\displaystyle \sum\limits_{\substack{|I|=k}}\left(\prod\limits_{\substack{i\in I \\ j\notin I}}\phi_{ji} \sum_{\substack{l\in I\\l\neq m}} g_{ml}\right)=\sum_{l\neq m}g_{ml}\left(\sum\limits_{\substack{|I|=k\\l\in I}}\prod\limits_{\substack{i\in I \\ j\notin I}}\phi_{ji}
\right)=-\sum_{l\neq m}g_{ml}u_l^{\{k\}}$$
we have
\beq\label{t65}
\begin{array}{c}
\displaystyle
 w^{\{k\}}_m=u^{\{k\}}_m\sum_{l\neq m}g_{lm}-\sum_{l\neq m}g_{ml}u_l^{\{k\}}=
u^{\{k\}}\sum_{l\neq m}f_{lm}=0\,.
 \end{array}
\eq
In the first equality we used (\ref{s26}), i.e. $u^{\{k\}}_m=u^{\{k\}}_l=u^{\{k\}}$. Finally, using (\ref{t14}) we obtain the last equality of (\ref{t65}).
\end{proof}

\section{Spin chain in the Ruijsenaars form}\label{sect5}
\setcounter{equation}{0}

As was mentioned in the Introduction the Ruijsenaars-Macdonald operators  were originally introduced in the form with square roots (\ref{Macd2}), which we called the Ruijsenaars form. Their spin XYZ generalization (\ref{spMacd2}) is obtained similarly to (\ref{spMacd0}) \cite{MZ}.
Let us write down these operators explicitly:
$$
  \displaystyle{
    {\mathcal D}'_k=\sum\limits_{1\leq i_1<...<i_k\leq N}\left(\!\prod\limits^{N}_{\hbox{\tiny{$ \begin{array}{c}{ j=1 }\\{ j\!\neq\! i_1...i_{k-1} } \end{array}$}}}\!\phi(z_j-z_{i_1})\cdots\phi(z_j-z_{i_k})\right)^{\frac{1}{2}}\left(
   \overleftarrow{\prod\limits_{j_1=1}^{i_1-1}} \bar{R}_{j_1 i_1}
    \ldots\
 \overleftarrow{\prod\limits^{i_k-1}_{\hbox{\tiny{$ \begin{array}{c}{ j_k=1 }\\{ j_k\!\neq\! i_1...i_{k-1} } \end{array}$}}}} \bar{R}_{j_k i_k}
  \right)\times
   }
  $$
  \beq\label{q20}
     \displaystyle{
       \times p_{i_1}\cdots p_{i_k}\times\left(
   \overrightarrow{\prod\limits^{i_k-1}_{\hbox{\tiny{$ \begin{array}{c}{ j_k\!=\!1 }\\{ j_k\!\neq\! i_{1}...i_{k-1}} \end{array}$}}}}\bar{R}_{i_k j_k}
   \ldots\
  \overrightarrow{\prod\limits^{i_{1}-1}_ {j_1=1}} \bar{R}_{i_{1} j_{1}}\right)
 \left(\!\prod\limits^{N}_{\hbox{\tiny{$ \begin{array}{c}{ j=1 }\\{ j\!\neq\! i_1...i_{k-1} } \end{array}$}}}\!\phi(z_{i_1}-z_j)\cdots\phi(z_{i_k}-z_{j})\right)^{\frac{1}{2}}\,.
 }
\eq
The operators (\ref{q20}) mutually commute and the freezing trick can be applied as well. Below we briefly describe it.

 The Hamiltonians of the classical elliptic Ruijsenaars-Schneider model in the Ruijsenaars form are given by:
\begin{equation}\label{t70}
 h_k'=\sum\limits_{\substack{|I|=k}}\prod\limits_{\substack{i\in I \\ j\notin I}}\sqrt{\phi(z_j-z_i)}\prod_{i\in I}e^{-v_i/c}\prod\limits_{\substack{i\in I \\ j\notin I}}\sqrt{\phi(z_i-z_j)}\qquad k=1,\dots,N\,,
\end{equation}
Each Hamiltonian $h_k'$ provides its dynamics through
the Hamiltonian equations:
\begin{equation}\label{t71}
\displaystyle{
  \frac{dz_j}{dt_k'}=\{h_k',z_j\}=\frac{\p h_k'}{\p v_j}\,,\qquad \frac{dv_j}{dt_k'}=\{h_k',v_j\}=-\frac{\p h_k'}{\p z_j}\,.
  }
\end{equation}
Here we use notations  similar to (\ref{s24}) and (\ref{s25}) for the set of classical velocities and for the first derivative of momenta (restricted to equilibrium points (\ref{s23})):
\begin{equation}\label{t72}
\displaystyle{
  u'^{\{k\}}_m=c\frac{dz_m}{dt_k'}\Big|_{\rm eq}
  }
  \end{equation}
  and
  \begin{equation}\label{t73}
\displaystyle{
  \qquad
  w'^{\{k\}}_m=-\frac{dv_m}{dt'_k}\Big|_{\rm eq}\,.
  }
\end{equation}
Evaluation of (\ref{t72}) and (\ref{t73}) is performed by using Hamiltonian equations (\ref{t71}):
\beq\label{t74}
u'^{\{k\}}_m=-\sum\limits_{\substack{|I|=k\\ m\in I}}\prod\limits_{\substack{i\in I\\j\notin I}}\sqrt{\phi(x_j-x_i)\phi(x_i-x_j)}\,,
\eq
\beq\label{t75}
w'^{\{k\}}_m=\frac{1}{2}\sum\limits_{\substack{|I|=k}}\prod\limits_{\substack{i\in I\\j\notin I}}\sqrt{\phi(x_j-x_i)\phi(x_i-x_j)}\sum_{l\neq m}f(x_m-x_l)\,,
\eq
where $f(x)$ was defined in (\ref{t11}).
\begin{lemma}
For any $k,l,m=1,...,N$ the analogues of (\ref{s26}) and (\ref{s27}) holds true
\begin{equation}\label{t76}
\displaystyle{
  u'^{\{k\}}_m=u'^{\{k\}}_l
  }
\end{equation}
and
\begin{equation}\label{t77}
\displaystyle{
  w'^{\{k\}}_m=0\,.
  }
\end{equation}
\end{lemma}
Identity (\ref{t76}) can be proved similarly to (\ref{t604}).  Identity (\ref{t77}) holds due to (\ref{t14}).

For the quantum model we use expansion of spin operators in variable $\eta$ and look at the first order. As for operators in Macdonald form we obtain matrix valued functions free of differential operators
\beq\label{t770}
\tilde{H}'_k={\rm Id}\, (D_k')^{[1]}-(\mathcal{D}'_k)^{[1]}.
\eq
Due to (\ref{t76}) and (\ref{t77}) being restricted on points (\ref{s23}) the operators
\beq\label{t771}
H'_k=\tilde{H}'_k\Big|_{\rm eq}
\eq
commute. However, we obtain the same Hamiltonians $H'_k=H_k$ as in Section \ref{sect3}. Note that operators (\ref{q10}) and (\ref{q20}) has the same matrix items, the difference is at most a scalar factor before them. The same will be in the spin chain Hamiltonians, in the Macdonald form  the scalar part of some kind of summand is
$
 \prod\limits_{\substack{i\in I\\j\notin I}}\phi(x_j-x_i)
$
and in the Ruijsenaars form
$$
 \sqrt{\prod\limits_{\substack{i\in I\\j\notin I}}\phi(x_j-x_i)\prod\limits_{\substack{i\in I\\j\notin I}}\phi(x_i-x_j)}\,.
$$
The following identity
\beq\label{t78}
\displaystyle \prod\limits_{\substack{i\in I\\j\notin I}} \frac{\phi(x_j-x_i)}{\phi(x_i-x_j)}=
1\,
\eq
establishes the equivalence between these two scalar factors and thus the equality of Hamiltonians $H'_k=H_k$. To prove (\ref{t78}) multiply the numerator and the denominator by $\prod\limits_{\substack{i,k\in I\\i\neq k}} \phi(x_k-x_i)$, then
\beq\label{t79}
\displaystyle \prod\limits_{\substack{i\in I\\j\notin I}} \frac{\phi(x_j-x_i)}{\phi(x_i-x_j)}= \prod\limits_{\substack{i\in I}}\prod_{j\neq i} \frac{\phi(x_j-x_i)}{\phi(x_i-x_j)}\,.
\eq
Put $j'=2i-j \mod N$, the multipliers $\phi(x_j-x_i)$  and $\phi(x_i-x_{j'})$ are reduced, thus each product $\displaystyle\prod_{j\neq i} \frac{\phi(x_j-x_i)}{\phi(x_i-x_j)}$ in the r.h.s of (\ref{t79}) equals 1.

To sum up we obtain the following:

\begin{predl}
Being restricted to the point (\ref{s23}) the expressions (\ref{s39}) ${\tilde H}_k={\rm Id}\, D_k^{[1]}-\mathcal{D}_k^{[1]}$ and (\ref{t770}) $\tilde{H}'_k={\rm Id}\, (D_k')^{[1]}-(\mathcal{D}'_k)^{[1]}$ give equal sets of
Hamiltonians (\ref{s431})  and (\ref{t771}) $H_k=H'_k$.
\end{predl}


\section{Trigonometric models}\label{sect6}
\setcounter{equation}{0}

In the trigonometric limit ${\rm Im}(\tau)\rightarrow +\infty$ we have
$\vth(z)=2\exp(\frac{\pi\imath\tau}{4})\sin(\pi z)+O(\exp(\frac{9\pi\imath\tau}{4}))$. Then the trigonometric  limits of the function (\ref{a0962}) is as follows:
 \beq\label{a081}
 \begin{array}{c}
  \displaystyle{
\phi^{\rm trig}(z)=\phi^{\rm trig}(z,\hbar)=\pi\cot(\pi z)+\pi\cot(\pi \hbar)\,.
 }
 \end{array}
 \eq
Also, in this limit
 \beq\label{a0811}
 \begin{array}{c}
  \displaystyle{
\wp(\hbar)-\wp(z)\rightarrow
\phi^{\rm trig}(z)\phi^{\rm trig}(-z)=\frac{\pi^2}{\sin^2(\pi\hbar)}
-\frac{\pi^2}{\sin^2(\pi z)}\,.
 }
 \end{array}
 \eq
Trigonometric limits of elliptic $R$-matrix were studied in \cite{AHZ}. In order to get a finite (not divergent) answer in the limit one should first  perform a gauge transformation
 \beq\label{y051}
 \begin{array}{c}
  \displaystyle{
  {\bar R}_{12}^\hbar(z)\rightarrow G^{(1)}G^{(2)}{\bar R}_{12}^\hbar(z)(G^{(1)})^{-1}(G^{(2)})^{-1}={\widetilde R}_{12}^\hbar(z)\,,
  }
 \end{array}
 \eq
where $G\in\MatM$ is a special matrix depending on $\hbar$ (see \cite{AHZ}) and the notation $G^{(i)}$ means the matrix $G$ is in the $i$-th tensor component. It is important that this matrix is independent of the spectral parameter $z$. For this reason the gauged transformed spin operator
 \beq\label{y05}
 \begin{array}{c}
  \displaystyle{
  \widetilde{\mathcal D}_k=G^{(1)}\ldots G^{(N)}{\mathcal D}_k
  (G^{(1)})^{-1}\ldots(G^{(N)})^{-1}
  }
 \end{array}
 \eq
is equal to initial spin operator constructed by means of the gauge transformed $R$-matrices. Thus, if $R$-matrix has a finite limit, then so do
$\widetilde{\mathcal D}_k$. Moreover, the spin operators keep their mutual commutativity in the trigonometric limit since
 \beq\label{y052}
 \begin{array}{c}
  \displaystyle{
  G^{(1)}\ldots G^{(N)}[{\mathcal D}_k,{\mathcal D}_l]
  (G^{(1)})^{-1}\ldots(G^{(N)})^{-1}=[\widetilde{\mathcal D}_k,\widetilde{\mathcal D}_l]=0\,.
  }
 \end{array}
 \eq
Therefore, our construction of spin chains is valid for any trigonometric limit obtained in the above mentioned way. Below we give several main examples of trigonometric $R$-matrices obtained in \cite{AHZ} (see also \cite{KZ19} for a review).

\subsection{Trigonometric $R$-matrices}
Let us begin with ${\rm GL}_2$ (i.e. $M=2$) trigonometric $R$-matrices. They
are given by 6-vertex XXZ $R$-matrix and its 7-vertex deformation \cite{Chered2}:
 \beq\label{y053}
 \begin{array}{c}
 R_{12}^\hbar(z)=\left(
 \begin{array}{cccc}
 \pi\cot(\pi \hbar)+\pi\cot(\pi z) & 0 & 0 & 0
 \\
 0 & \displaystyle{ \frac{\pi}{\sin(\pi\hbar)} } &  \displaystyle{ \frac{\pi}{\sin(\pi z)} }  & 0
  \\
 0 &  \displaystyle{ \frac{\pi}{\sin(\pi z)} }  & \displaystyle{ \frac{\pi}{\sin(\pi\hbar)} } & 0
 \\
 c_7\sin(\pi(\hbar+z)) & 0 & 0 & \pi\cot(\pi \hbar)+\pi\cot(\pi z)
 \end{array}
 \right)\,,
  \end{array}
 \eq
where $c_7$ is an arbitrary constant. When $c_7=0$ it is the ${\rm GL}_2$ XXZ $R$-matrix. The $R$-matrix (\ref{y053}) is normalized as
 \beq\label{y054}
 \begin{array}{c}
  \displaystyle{
  R_{12}^\hbar(z)R_{21}^\hbar(-z)=\phi^{\rm trig}(z)\phi^{\rm trig}(-z)\,{\rm Id}=\Big( \frac{\pi^2}{\sin^2(\pi \hbar)}-
  \frac{\pi^2}{\sin^2(\pi z)}\Big){\rm Id}\,,
  }
 \end{array}
 \eq
 so that in order to use it for construction of spin operators one should
 previously normalize it as given (\ref{q04})-(\ref{q05}) using (\ref{a081}).

For arbitrary $M$ we have the following set of trigonometric $R$-matrices. We write them in the standard basis $\{e_{ab}\}$ in ${\rm Mat}_M$, i.e.
 \beq\label{y055}
 \begin{array}{c}
  \displaystyle{
  R_{12}^\hbar(z)=\sum\limits_{a,b,c,d,=1}^M e_{ab}\otimes e_{cd} R^{\hbar}_{ab,cd}(z)\,.
  }
 \end{array}
 \eq
The first one is
$\mZ_N$-invariant $A_{M-1}$
trigonometric $R$-matrix \cite{Chered2,PerkKulish}:
  \beq\label{q200}
   \begin{array}{c}
   \displaystyle{
  (R_1)^{\hbar}_{ab,cd}(z)=
  }
  \\ \ \\
     \displaystyle{
  =\delta_{ab}\delta_{cd}\delta_{ac}\,\pi\Big(\cot(\pi\hbar)+\cot(\pi z)\Big)+\delta_{ab}\delta_{cd}\pi\,\varepsilon(a\neq c)\,\frac{ \exp\Big(\frac{\pi\imath\hbar}{M}\Big(2(a-c)-M{\rm sign}(a-c)\Big)\Big)}{\sin(\pi\hbar)}+
  }
  \\ \ \\
   \displaystyle{
 +\delta_{ad}\delta_{bc}\pi\,\varepsilon(a\neq c)\,\frac{ \exp\Big(\frac{\pi\imath z}{M}\Big(2(a-c)-M{\rm sign}(a-c)\Big)\Big)}{\sin(\pi z)}\,,
  }
  \end{array}
  \eq
  where we use notation
  \beq\label{q201}
  \begin{array}{c}
  \displaystyle{
 \varepsilon(\hbox{A})=\left\{\begin{array}{l} 1\,,\hbox{if A is true}\,,\\ 0\,,\hbox{if A is false}\,.\end{array}\right.
 }
 \end{array}
 \eq
 The next $R$-matrix is the Baxterization of the
(trigonometric) Cremmer-Gervais $R$-matrix \cite{Babelon,AHZ}. It is
obtained as a sum of (\ref{q200}) and additional term:
  \beq\label{q202}
   \begin{array}{c}
   \displaystyle{
  (R_2)^{\hbar}_{ab,cd}(z)=(R_1)^{\hbar}_{ab,cd}(z)-
  }
  \\ \ \\
     \displaystyle{
  -2\pi\imath\,\delta_{a+c,b+d}
  \Big( \varepsilon(a<b<c)- \varepsilon(c<b<a) \Big)
  \exp\Big(\frac{\pi\imath z}{M}(a-b)+\frac{\pi\imath \hbar}{M}(d-c)\Big)\,.
  }
  \end{array}
  \eq
  Finally, the most general is the non-standard $R$-matrix \cite{AHZ}:
  \beq\label{q203}
   \begin{array}{c}
   \displaystyle{
  (R_3)^{\hbar}_{ab,cd}(z)=(R_2)^{\hbar}_{ab,cd}(z)-
  }
  \\ \ \\
     \displaystyle{
  -2\pi\imath\,\delta_{a+c,b+d+M}
  \Big( \delta_{aM}
  \exp\Big(-\frac{\pi\imath z}{M}b-\frac{\pi\imath \hbar}{M}d\Big)
  -\delta_{cN}
  \exp\Big(\frac{\pi\imath z}{M}d+\frac{\pi\imath \hbar}{M}b\Big)
  \Big)\,.
  }
  \end{array}
  \eq
All the above $R$-matrices are  normalized as given in (\ref{y054}).

 Let us also remark on the widely known XXZ $R$-matrix \cite{Jimbo} for the affine quantized
algebra ${\hat{\mathcal U}}_q({\rm gl}_M)$:
  \beq\label{q270}
   \begin{array}{c}
   \displaystyle{
  R^\hbar_{12}(z)=
   \pi\Big(\cot(\pi z)+\coth(\pi\hbar)\Big)\sum\limits_{i=1}^M e_{ii}\otimes e_{ii}+
  }
  \\ \ \\
   \displaystyle{
 +\frac{\pi}{\sin(\pi\hbar)}\sum\limits_{i\neq j}^M e_{ii}\otimes e_{jj}+
 \frac{\pi}{\sin(\pi z)}\sum\limits_{i< j}^M
 \Big( e_{ij}\otimes e_{ji}\,e^{\pi\imath z}+e_{ji}\otimes
 e_{ij}\,e^{-\pi\imath z}\Big)\,.
  }
  \end{array}
  \eq
 In the $M=2$ case this yields ${\hat{\mathcal U}}_q({\rm gl}_2)$ $R$-matrix underlying
  the q-deformed Haldane-Shastry model \cite{Uglov}. The Uglov-Lamers Hamiltonian \cite{Lam,LPS} is reproduced in the next subsection. Here we mention that (\ref{q270}) for $M=2$
  is (\ref{y053}) with $c_7=0$ up to simple gauge transformation. However, in this case the gauge transformation
  ${\bar R}_{12}^\hbar(z-w)\rightarrow G^{(1)}(z)G^{(2)}(w){\bar R}_{12}^\hbar(z-w)(G^{(1)})^{-1}(z)(G^{(2)})^{-1}(w)$
  depends on the spectral parameter: $G(z)={\rm diag}(e^{\pi\imath z/2},e^{-\pi\imath z/2})$. Therefore, the arguments from the beginning of the section do not work in this case. Notice that a generic gauge transformation
  is not applicable since the resultant $R$-matrix may not depend on the difference of spectral parameters (one may try to extend results of \cite{MZ} to generic $R$-matrices $R_{12}(z,w)\neq R_{12}(z-w)$). However, in the case of study the gauge transformation $G(z)$ appears to be admissible.
  To see it we prove in the Appendix that $R$-matrix (\ref{q270}) provides commutative set of spin operators (\ref{q10}) and commuting set of spin chain Hamiltonians as was described in the previous sections.

  It is an interesting problem to find all admissible gauge transformations for $R$-matrices, which keep commutativity of spin operators (\ref{q10}) and commutativity of spin chain Hamiltonians. This question deserves further elucidation.

\subsection{Uglov-Lamers Hamiltonian for q-deformed Haldane-Shastry model}
The goal of this paragraph is to reproduce the first Hamiltonian in the trigonometric case and to show that it
coincides with the Uglov-Lamers Hamiltonian of the q-deformed Haldane-Shastry model in the form suggested in \cite{Lam,LPS}.

We use multiplicative notations
\beq\label{q31}
u= e^{2\pi\imath z}\,, \qquad t= e^{2\pi\imath \hbar}\,.
\eq
Denote by $a(u)$ the analogue of (\ref{a081}) in multiplicative variables:
 \beq\label{q31a}
 a(u)=\phi^{\rm trig}(z)=\pi\imath\left(\frac{t+1}{t-1}+\frac{u+1}{u-1}\right)=2\pi \imath \frac{u t-1}{(u-1)(t-1)} \,.
 \eq
 We use notations similar to (\ref{s091}) for the subset $I$ of the set $\{1,\dots,N\}$:
 \beq\label{q31b}
 \begin{array}{c}
  \displaystyle{
A_I=\prod\limits_{\substack{i\in I \\ j\notin I}}a\left(\frac{y_j}{y_i}\right)
=\prod\limits_{\substack{i\in I \\ j\notin I}}\frac{2\pi \imath}{t-1}\frac{t y_j-y_i}{y_j-y_i}\,.
  }
 \end{array}
\eq
In $M=2$ case denote by  $ R_{12}^{\rm{trig}}(u)$  the  XXZ  $R$-matrix (\ref{q270}) in notations (\ref{q31}):
\beq\label{q32}
 \begin{array}{c}
 R_{12}^{\rm{trig}}(u)=\pi \imath \left(
 \begin{array}{cccc}
 \displaystyle{\frac{t+1}{t-1}+\frac{u+1}{u-1} }& 0 & 0 & 0
 \\
 0 & \displaystyle{\frac{2 t^{1/2}}{t-1} }& \displaystyle{\frac{2u}{u-1} } & 0
  \\
 0 &  \displaystyle{ \frac{2}{u-1}}  &\displaystyle{ \frac{2 t^{1/2}}{t-1} }  & 0\\
0 & 0 & 0 & \displaystyle{\frac{t+1}{t-1}+\frac{u+1}{u-1}}
 \end{array}
 \right)\,
  \end{array}
\eq
and by  $\bar{R}_{12}^{\rm{trig}}(u)$\footnote{This R-matrix coincides with (1.23) from \cite{LPS} up to the permutation operator $\check{R}_{12}(u)=\bar{R}_{12}^{\rm trig}(u) P_{12} $} its normalized (\ref{q05}) version :
\beq\label{q33}
 \begin{array}{c}
 \bar{R}_{12}^{\rm{trig}}(u)=\left(
 \begin{array}{cccc}
 1& 0 & 0 & 0
 \\
 0 & \displaystyle{\frac{t^{1/2}(u-1)}{ut-1} }& \displaystyle{\frac{u(t-1)}{ut-1} } & 0
  \\
 0 &  \displaystyle{ \frac{t-1}{u t-1}}  &\displaystyle{ \frac{t^{1/2}(u-1)}{ut-1} } & 0 \\
0 & 0 & 0 & 1
 \end{array}
 \right)\,.
  \end{array}
\eq
The quantum Yang-Baxter equation (\ref{QYB}) takes the form
\beq\label{q34}
\begin{array}{c}
\displaystyle{
   \bar{R}^{\rm{trig}}_{12}(u) \bar{R}^{\rm{trig}}_{13}(uv) \bar{R}^{\rm{trig}}_{23}(v) =
      \bar{R}^{\rm{trig}}_{23}(v) \bar{R}^{\rm{trig}}_{13}(uv) \bar{R}^{\rm{trig}}_{12}(u)
      }\,,
\end{array}\eq
and  the unitarity property (\ref{q05}) is written as
\beq\label{q35}
    {\displaystyle {\bar R}^{\rm{trig}}_{12}(u) {\bar R}^{\rm{trig}}_{21}\left(\frac{1}{u}\right)= {\rm Id}\,.}
\eq
Notice one important property: expression ${\bar R}^{\rm{trig}}_{12}\left(\frac{1}{u}\right)u\frac{\p}{\p u}  {\bar R}^{\rm{trig}}_{21}(u)$
is proportional to matrix $C_{12}$ which does not depend on the variable $u$:
\beq \label{q36}
 {\bar R}^{\rm{trig}}_{12}\left(\frac{1}{u}\right)u\frac{\p}{\p u}  {\bar R}^{\rm{trig}}_{21}(u)
 =\frac{(1-t)u}{(t-u)(tu-1)} C_{12}\,.
\eq
Here
\beq\label{q360}
C_{12}=\left(
 \begin{array}{cccc}
 0& 0 & 0 & 0
 \\
 0 & 1& -\sqrt{t} & 0
  \\
 0 &  -\sqrt{t}& t & 0 \\
0 & 0 & 0 & 0
 \end{array}
 \right)\,.
\eq
We denote by $C_{ij}$ the corresponding operator acting  in the $i$-th  and $j$-th tensor components of $\mathcal H$.

Now we are going to compute the first Hamiltonian in the trigonometric case. We use the notations
$H^{\rm{trig}}_k$ for the $k$-th Hamiltonians and
$\tilde{H}^{\rm{trig}}_k$  for the corresponding  matrix-valued function from (\ref{s39}), which after restriction to
\beq\label{q36a}
{\rm eq}:\ y_j\to \exp\left(\frac{2\pi \imath }{N}j\right)
\eq
turns into the Hamiltonian:
\begin{equation}\label{q36b}
\displaystyle{
 H^{\rm{trig}}_k={\tilde H}^{\rm{trig}}_k\Big|_{\rm eq}\,.
  }
\end{equation}
In the trigonometric limit  ${\bf H}_1$ from (\ref{s461}) in multiplicative notations (\ref{q31}) takes the form:
\begin{equation}\label{q41}
\begin{array}{c}
{\displaystyle
{\bf H}^{\rm{trig}}_1=2\pi\imath\sum_{k<i}\bar{R}^{\rm{trig}}_{i-1,i}\left(\frac{y_{i-1}}{y_i}\right)\dots \bar{R}^{\rm{trig}}_{k+1,i}
\left(\frac{y_{k+1}}{y_i}\right)\times}
\\ \ \\
{\displaystyle
\times\bar{R}^{\rm{trig}}_{k,i}\left(\frac{y_{k}}{y_i}\right)\left(y_i\frac{\partial}{\partial y_i}\bar{R}^{\rm{trig}}_{i,k}\left(\frac{y_{i}}{y_k}\right)\right)\bar{R}^{\rm{trig}}_{i,k+1}\left(\frac{y_{i}}{y_{k+1}}\right)\dots \bar{R}^{\rm{trig}}_{i,i-1}\left(\frac{y_{i}}{y_{i-1}}\right)\Big|_{\rm eq}\,.
}
\end{array}
\end{equation}
Using (\ref{q36}) we obtain the expression for the first Hamiltonian of
the q-deformed Haldane-Shastry model or the Uglov-Lamers Hamiltonian\footnote{The expression (\ref{q40}) coincides (up to a constant factor)  with (1.20) in \cite{LPS} taking into account difference in notations of $R$-matrices.}:
\begin{equation}\label{q40}
\begin{array}{c}
{\displaystyle
{\bf H}^{\rm{trig}}_1=2\pi\imath(1-t)\sum_{k<i}
\frac{y_i y_k}{(t y_k-y_i)(t y_i-y_k)}}\times
\\ \ \\
{\displaystyle
\times\bar{R}^{\rm{trig}}_{i-1,i}\left(\frac{y_{i-1}}{y_i}\right)\dots \bar{R}^{\rm{trig}}_{k+1,i}
\left(\frac{y_{k+1}}{y_i}\right)C_{ik}\bar{R}^{\rm{trig}}_{i,k+1}\left(\frac{y_{i}}{y_{k+1}}\right)\dots \bar{R}^{\rm{trig}}_{i,i-1}\left(\frac{y_{i}}{y_{i-1}}\right)\Big|_{\rm eq}\,.
}
\end{array}
\end{equation}

In the limit $t\to 1$ we have  $\bar{R}^{\rm{trig}}_{ij}\to {\rm Id}$ and  $C_{ij}\to (1-P_{ij})$, thus
\beq\label{q400}
\lim_{t\to 1}  \frac{{\bf H}^{\rm{trig}}_1}{(1-t)}=2\pi\imath\sum_{k<i}
\frac{-y_i y_k}{(y_k-y_i)^2}(1-P_{ik})\Big|_{\rm eq}\,.
\eq
The right hand side  of (\ref{q400}) coincides with the Hamiltonian of Haldane-Shastry spin chain (\ref{s01}) up to a constant factor.

To summarize, a straightforward trigonometric limit of the elliptic model provides the models based on $R$-matrices
(\ref{y053}) for $M=2$ and (\ref{q203}) in the general case. In the non-relativistic limit these models turn into
anisotropic Haldane-Shastry models (see (\ref{s0911}) below). At the same time the Uglov's q-deformed Haldane-Shastry model provides the isotropic non-relativistic limit (\ref{s01}), i.e. the standard Haldane-Shastry model.
The underlying $R$-matrix is given by (\ref{q32}) and (\ref{q270}) for higher rank case.
In order to include these models into consideration we explain in the Appendix C that the spin operators
(\ref{q10}) with the ${\rm U}_q({\widehat {\rm  gl}_M})$ $R$-matrix (\ref{q270}) mutually commute. Then the freezing trick yields the Uglov's q-deformed long-range spin chains.

To avoid confusion notice that in this paper we use the usual notations $t$ and $q$ for Macdonald parameters. At the same time the standard term "q-deformed" means deformation with parameter $t$. In our case it would be more correctly to say "$t^{1/2}$-deformed".

\paragraph{Remark.} In the trigonometric case the analogue of Lemma \ref{lemma1} holds and the corresponding sums of the coefficients (\ref{q31b}) are equal to $q$-binomial coefficients with $q=t^{1/2}$ (see \cite{Uglov,LPS}):
 \beq \label{q39}
\sum\limits_{\substack{|I|=k\\ l\in I}} A_I\Big|_{\rm eq}=
\sum\limits_{\substack{|I'|=k\\ m\in I'}} A_{I'}\Big|_{\rm eq}=\left(\frac{2\pi\imath }{t^{-1/2}-t^{1/2}}\right)^{k(N-k)}\frac{k}{N}
\left[{\begin{array}{c}N\\k\\\end{array}}\right],
\eq
where we use the notations for $q$-binomial coefficients
\beq\label{qbinom}
[n]:=[n]_{\sqrt{t}}=\frac{t^{n/2}-t^{-n/2}}{t^{1/2}-t^{-1/2}}\,,\qquad
[n]!:=[n] [n-1]\dots[2]\,,\qquad
\left[{\begin{array}{c}n\\k\\\end{array}}\right] =\frac{[n]!}{[n-k]![k]!}\,.
\eq
The latter means that the sum in the left hand side  of (\ref{q39}) with elliptic
coefficients $A_I$ can be used for definition of the elliptic version for
q-binomial coefficients. It would be interesting to compare such definition with that one suggested in \cite{Schlos}.

In the end of the Section let us also make a remark on the rational limit, which we do not discuss in this paper.  In this case  $\phi^{\rm rat}(z,u)=1/z+1/u$. The rational limits of $R$-matrix are also known. It is the Yang's $R$-matrix $R_{12}^\hbar(z)=\hbar^{-1}{\rm Id}+z^{-1}P_{12}$ and its deformations (11-vertex one in $M=2$ case and higher analogues for $M>2$). In the Yang's case this provide isotropic spin Ruijsenaars-Macdonald Hamiltonians.
However, there is a problem with proceeding to long-range spin chain.
As was mentioned in the Introduction the Polychronakos-Frahm chain is related to the rational spin Calogero-Moser model in the harmonic potential. The latter is necessary to have a finite set of equilibrium positions \cite{Calogero3}. The relativistic generalization, which we study does not provide the harmonic potential in the non-relativistic limit. That is, in order to study the rational case we need to extend our spin Ruijsenaars-Macdonald operators
to the case of BC root system, which includes relativistic analogue of oscillator terms. This case will be studied elsewhere.

\section{Limit to non-relativistic models}\label{sect7}
\setcounter{equation}{0}
 The coupling constant in the set of XYZ spin Ruijsenaars-Macdonald operators (\ref{q10}) and in the set of the spin chain Hamiltonians $H_i$ (\ref{s431}) is the parameter $\hbar$. The non-relativistic limit $\hbar\rightarrow 0$ is similar to transition to the Calogero-Moser model from the Ruijsenaars-Schneider one.
 This parameter plays the role of the Planck constant in a quantum $R$-matrix entering the definition of Hamiltonians. Decomposition of quantum $R$-matrix in $\hbar\rightarrow 0$ is the classical limit (\ref{r052}). Its first nontrivial term is the classical $r$-matrix.

 \paragraph{Classical $r$-matrix.} Notice that we deal with the normalized $R$-matrix ${\bar R}^\hbar_{ij}$ as in (\ref{q04})-(\ref{q05}). In the trigonometric case one should use normalization by (\ref{a081}). While non-normalized $R$-matrix (e.g. elliptic (\ref{BB}) or any of trigonometric $R$-matrices from the subsection 6.1.) has
 the classical limit as given in (\ref{r052}), a normalized $R$-matrix
 has this limit in the form:
 \beq\label{q60}
 \displaystyle{
 {\bar R}^\hbar_{12}(z)={\rm Id}+\hbar\, {\bar r}_{12}(z)+
 \hbar^2 {\bar m}_{12}(z)+O(\hbar^3)\,.
 }
 \eq
 For the elliptic $R$-matrix we have the following classical $r$-matrix
 \beq\label{q61}
 \displaystyle{
 {\bar r}_{12}(z)= {\rm Id}\, \frac{1-M}{M}\,E_1(z)+\frac{1}{M}
    \sum_{\al \neq 0} \varphi_\al (z, \om_\al) T_\al \otimes T_{-\al}\,,
 }
 \eq
 satisfying the classical Yang-Baxter equation (\ref{r054}). For example, in $M=2$ case it is as follows:
 \beq\label{r7241}
 \begin{array}{c}
  \displaystyle{
 {\bar r}_{12}(z)
 =-\frac{{\rm Id} _{4\times 4}}{2}\,E_1(z)+
 \frac{1}{2}\left(
 \begin{array}{cccc}
 \ti\vf_{10} & 0 & 0 & \ti\vf_{01}-\ti\vf_{11}
 \\
 0 & -\ti\vf_{10} & \ti\vf_{01}+\ti\vf_{11} & 0
  \\
 0 & \ti\vf_{01}+\ti\vf_{11} & -\ti\vf_{10} & 0
 \\
 \ti\vf_{01}-\ti\vf_{11} & 0 & 0 & \ti\vf_{10}
 \end{array}
 \right)\,,
  }
 \end{array}
\eq
 where
 \beq\label{r8261}
 \begin{array}{c}
  \displaystyle{
 \ti\vf_{10}=\phi(z,\frac{1}{2})\,,\quad
 \ti\vf_{01}=e^{\pi\imath z}\phi(z,\frac{\tau}{2})\,,\quad
 \ti\vf_{11}=e^{\pi\imath z}\phi(z,\frac{1+\tau}{2})\,.
 }
  \end{array}
 \eq
 The $r$-matrix (\ref{q61}) differs from the one (\ref{r053}) by a scalar term only:
 \beq\label{q62}
 \displaystyle{
 r_{12}(z)={\rm Id}\, E_1(z) + {\bar r}_{12}(z)\,.
 }
 \eq
In what follows we also need the derivative of (\ref{q61})
with respect to the spectral parameter\footnote{In $M=2$ case it is also helpful
to use relations $\p_z\ti\vf_\al=-\ti\vf_\be\ti\vf_\ga$ for any distinct $\al,\be,\ga\in\{10,01,11\}$.}
$\p {\bar r}_{12}(z)$:
 \beq\label{q63}
 \displaystyle{
 \p_z {\bar r}_{12}(z)
 \stackrel{(\ref{a0965})}{=}{\rm Id}\,\frac{M-1}{M}E_2(z)+
 \frac{1}{M}
    \sum_{\al \neq 0}
    \varphi_\al (z, \om_\al)\Big(\frac{2\pi\imath \al_2}{M}+E_1(z+\om_\al)-E_1(z)\Big)
    T_\al \otimes T_{-\al}\,,
 }
 \eq
 where $E_2(z)=-\p_zE_1(z)$.
Similarly to (\ref{r056}) we have the following parity properties:
 \beq\label{q64}
 \begin{array}{c}
    \displaystyle{
    {\bar r}_{12}(z)=-{\bar r}_{21}(-z)\,,\qquad
    {\bar m}_{12}(z)={\bar m}_{21}(-z)\,,\qquad
    \p {\bar r}_{12}(z)=\p {\bar r}_{21}(-z)\,.
    }
\end{array}
\eq

 \paragraph{The limit from the first Hamiltonian.}  Consider the limit $\hbar\rightarrow 0$ of the first Hamiltonian (\ref{s461}). Plugging expansion (\ref{q60}) and its derivative
${\bar F}_{ij}^\hbar=\hbar\p {\bar r}_{ij}+\hbar^2\p {\bar m}_{ij}+O(\hbar^3)$ into
(\ref{s461}) one gets:
 \beq\label{q65}
 \begin{array}{c}
    \displaystyle{
    {\bf H}_1=\hbar\sum\limits^N_{i>j}\p {\bar r}_{ij}+
    \hbar^2\Big(
    \sum\limits^N_{i>j} ({\bar r}_{ji}\p{\bar r}_{ij} +\p {\bar m}_{ij}) + \sum\limits^N_{i<j<k} [{\bar r}_{jk}, \p{\bar r}_{ki}] \Big)+O(\hbar^3)\,,
    }
\end{array}
\eq
where $\p {\bar r}_{ab}=\p {\bar r}_{ab}(x_a-x_b)$ and similarly for
${\bar r}_{ab}$ and ${\bar m}_{ab}$.
Denote expansion (\ref{q65}) as
 \beq\label{q651}
 \begin{array}{c}
    \displaystyle{
    {\bf H}_1=\hbar {\bf H}_1^{(1)}+
    \hbar^2{\bf H}_1^{(2)}+O(\hbar^3)\,,
    }
    \\ \ \\
        \displaystyle{
    {\bf H}_1^{(1)}=\sum\limits^N_{i>j}\p {\bar r}_{ij}\,,\qquad
    {\bf H}_1^{(2)}=\sum\limits^N_{i>j} ({\bar r}_{ji}\p{\bar r}_{ij} +\p {\bar m}_{ij}) + \sum\limits^N_{i<j<k} [{\bar r}_{jk}, \p{\bar r}_{ki}]\,.
    }
\end{array}
\eq
The first nontrivial Hamiltonian of the non-relativistic spin chain is
as follows\footnote{The numeration is shifted by 1 likewise it happens in many-body systems, where the first nontrivial Hamiltonian of the Calogero-Moser model is the second one, while the first is the sum of momenta.}:
 \beq\label{q66}
 \begin{array}{c}
    \displaystyle{
    {\mathcal H}_2={\bf H}_1^{(1)}=\sum\limits^N_{i>j}\p {\bar r}_{ij}(x_i-x_j)\,.
    }
\end{array}
\eq

\paragraph{The limit from the second Hamiltonian.}
Next, consider the limit of the second Hamiltonian (\ref{s472}). More precisely, let us consider the limit of $\hbar^{-2}{\bf H}_2$. Notice that
(since $\wp(\hbar)=\hbar^{-2}+O(\hbar^2)$)
 \beq\label{q67}
 \begin{array}{c}
    \displaystyle{
    \frac{1}{\hbar^2}\frac{1}{\wp(\hbar)-\wp(x_i-x_j)}=
    \frac{1}{1-\hbar^2\wp(x_i-x_j)+O(\hbar^4)}\,.
    }
\end{array}
\eq
Therefore, the coefficients $1/(\wp(\hbar)-\wp(x_i-x_j))$ do not provide any input into the leading and the next to leading order in $\hbar$ of ${\bf H}_2$. Similarly to
(\ref{q651}) we have the expansion in the form:
 \beq\label{q68}
 \begin{array}{c}
    \displaystyle{
    \hbar^{-2}{\bf H}_2=\hbar {\bf H}_2^{(1)}+
    \hbar^2{\bf H}_2^{(2)}+O(\hbar^3)\,,
    }
\end{array}
\eq
where
 \beq\label{q69}
 \begin{array}{c}
    \displaystyle{
    {\bf H}_2^{(1)}=\sum\limits_{i<k<l}^N
    \Big(\p{\bar r}_{ki}+\p{\bar r}_{li}+\p{\bar r}_{lk}\Big)
    =(N-2)\sum\limits_{i>j}^N\p{\bar r}_{ij}
    }
\end{array}
\eq
and
 \beq\label{q70}
 \begin{array}{c}
    \displaystyle{
    {\bf H}_2^{(2)}=(N-2)\sum\limits^N_{i>j} ({\bar r}_{ji}\p{\bar r}_{ij} +\p {\bar m}_{ij})+(N-3)\sum\limits^N_{i<j<k} [{\bar r}_{jk}, \p{\bar r}_{ki}]+
    \sum\limits^N_{i<j<k} [{\bar r}_{ij}, \p{\bar r}_{ki}]\,,
    }
\end{array}
\eq
and the above argument with (\ref{q67}) means that the coefficients
$1/(\wp(\hbar)-\wp(x_i-x_j))$ do not effect (\ref{q69}) and (\ref{q70}).

Now we use commutativity of ${\bf H}_1$ and ${\bf H}_2$:
 \beq\label{q71}
 \begin{array}{c}
    \displaystyle{
    [{\bf H}_1,\hbar^{-2}{\bf H}_2]=0=\hbar^2[{\bf H}_1^{(1)},{\bf H}_2^{(1)}]+\hbar^3\Big( [{\bf H}_1^{(1)},{\bf H}_2^{(2)}]+[{\bf H}_1^{(2)},{\bf H}_2^{(1)}] \Big)+O(\hbar^4)\,.
    }
\end{array}
\eq
From (\ref{q66}) and (\ref{q69}) we conclude that ${\bf H}_2^{(1)}=(N-2){\bf H}_1^{(1)}$, so that $\hbar^2$ term in the right hand side  of (\ref{q71}) vanish. Then
 \beq\label{q72}
 \begin{array}{c}
    \displaystyle{
     [{\bf H}_1^{(1)},{\bf H}_2^{(2)}-(N-2){\bf H}_1^{(2)}]=
     [{\mathcal H}_2,{\bf H}_2^{(2)}-(N-2){\bf H}_1^{(2)}]=0\,.
    }
\end{array}
\eq
Thus, we define the next non-relativistic Hamiltonian as
 \beq\label{q73}
 \begin{array}{c}
    \displaystyle{
     {\mathcal H}_3={\bf H}_2^{(2)}-(N-2){\bf H}_1^{(2)}=
      \sum\limits^N_{i<j<k} [{\bar r}_{ij}(x_i-x_j)+{\bar r}_{kj}(x_k-x_j), \p{\bar r}_{ki}(x_k-x_i)]\,,
    }
\end{array}
\eq
and
 \beq\label{q731}
 \begin{array}{c}
    \displaystyle{
     [{\mathcal H}_2,{\mathcal H}_3]=0\,.
    }
\end{array}
\eq
Notice also that by differentiating the classical Yang-Baxter equation (\ref{r054}) one gets
 \beq\label{q74}
 \begin{array}{c}
    \displaystyle{
      [{\bar r}_{ki}+{\bar r}_{kj}, \p{\bar r}_{ij}]=
      [{\bar r}_{jk}+{\bar r}_{ji}, \p{\bar r}_{ki}]=
       [{\bar r}_{ij}+{\bar r}_{ik}, \p{\bar r}_{jk}]\,.
    }
\end{array}
\eq
The latter relations together with the properties (\ref{q64}) provide a certain freedom
in the definition of ${\mathcal H}_3$  (\ref{q73}). Also, in definitions of the Hamiltonians (\ref{q66}), (\ref{q73}) one can use the $r$-matrix $r_{ij}(z)$ (\ref{q62}) instead of ${\bar r}_{ij}(z)$ since they differ by a term proportional to identity operator, so that it does not effect the commutativity of ${\mathcal H}_2$ and ${\mathcal H}_3$. In such a form these
Hamiltonians were obtained in \cite{SeZ} (see also \cite{Z18}) using a different approach based on $R$-matrix valued Lax pairs. The model given by ${\mathcal H}_2$ can be viewed as a result of the freezing trick applied to the model of interacting tops \cite{LZ,LOSZ,GZ}.

To summarize, we proved the following
\begin{predl}
 The non-relativistic limit of the first (\ref{s46})-(\ref{s461}) and the
 second (\ref{s47}),(\ref{s472}) Hamiltonians of q-deformed spin chain
 provides a pair of commuting Hamiltonians ${\mathcal H}_2$ (\ref{q66}) and ${\mathcal H}_3$ (\ref{q73}).
\end{predl}

\paragraph{Trigonometric models.} In the trigonometric case it is possible to use all
$R$-matrices discussed in the previous Section. By substituting the corresponding classical $r$-matrix into (\ref{q66}) and (\ref{q73})
one obtains commuting Hamiltonians. For example when $M=2$ we deal with $R$-matrix (\ref{y053}). Its classical limit provides the following classical $r$-matrix:
 \beq\label{y0531}
 \begin{array}{c}
 {r}_{12}(z)=\left(
 \begin{array}{cccc}
 \pi\cot(\pi z) & 0 & 0 & 0
 \\
 0 & 0 &  \displaystyle{ \frac{\pi}{\sin(\pi z)} }  & 0
  \\
 0 &  \displaystyle{ \frac{\pi}{\sin(\pi z)} }  & 0 & 0
 \\
 c_7\sin(\pi z) & 0 & 0 & \pi\cot(\pi z)
 \end{array}
 \right)\,,
  \end{array}
 \eq
 and ${\bar r}_{12}(z)=r_{12}(z)-{\rm Id}\,\pi\cot(\pi z)$.
 Plugging it into (\ref{q66}) yields the Hamiltonian
 \beq\label{s0911}
 \begin{array}{c}
  \displaystyle{
 H^{\rm{7v}}=\frac{\rm Id}{2}\sum\limits_{i\neq j}^N\frac{\pi^2}{\sin^2(\pi(x_i-x_j))}-
 }
 \\ \ \\
   \displaystyle{
  -\frac{\pi^2}{2}\sum\limits_{i\neq j}^N \left(
  \frac{\cos(\pi(x_i-x_j))(\sigma_1^{(i)}\sigma_1^{(j)}+
  \sigma_2^{(i)}\sigma_2^{(j)})+\sigma_3^{(i)}\sigma_3^{(j)}
   }{\sin^2(\pi(x_i-x_j))} -\frac{c_7}{\pi}\,\cos(\pi(x_i-x_j))\, \sigma_-^{(i)}\sigma_-^{(j)}\right)\,,
  }
 \end{array}
\eq
 where ''7v'' stands for 7-vertex and $\sigma_-=e_{21}=\mats{0}{0}{1}{0}$.
 When $c_7=0$ it reproduces (up to common constant and the scalar first term in (\ref{s0911})) the XXZ Hamiltonian (\ref{s09}) derived in \cite{HW} as a potential of anisotropic spin Calogero-Moser model. The corresponding long-range spin chain was described in \cite{SeZ}.

 Notice also that by using the 6-vertex $R$-matrix in a slightly different gauge (\ref{q32}) the non-relativistic limit provides
 the isotropic Haldane-Shastry model (\ref{q400}), (\ref{s01}).

\section{Conclusion}
\setcounter{equation}{0}

We proposed a set of mutually commuting Hamiltonians for elliptic ${\rm GL}_M$ generalization of the q-deformed anisotropic Haldane-Shastry long-range spin chain.
Summarizing auxiliary statements from Sections \ref{sect3} and \ref{sect4} we proved the following
\begin{theor}\label{th1}
Consider expansions $D_k=D_k^{[0]}+\eta D_k^{[1]}+O(\eta^2)$ (\ref{s29}) and $\mD_k=\mD_k^{[0]}+\eta \mD_k^{[1]}+O(\eta^2)$ (\ref{s34}) of the spinless
(\ref{Dscalar}) and the spin XYZ (\ref{q10})  Ruijsenaars-Macdonald operators respectively. Then $\mathcal{D}_k^{[1]}={\rm Id}\, D_k^{[1]}-{\tilde H}_k$
(\ref{s39}), where ${\tilde H}_k$ are ${\rm Mat}_M^{\otimes N}$-valued functions free of differential operators. Being restricted to the points $z_j=x_j=j/N$ (\ref{s431}), (\ref{s23}) these matrix-valued functions
turn into the set of mutually commuting Hamiltonians (\ref{s45}). The set of points is an equilibrium position in the underlying classical spinless Ruijsenaars-Schneider model (\ref{s26})-(\ref{s222}).
\end{theor}
The proof of this statement uses a set of elliptic function identities from Section \ref{sect4}.

In Section \ref{sect5} we proved
\begin{theor}\label{th2}
The Theorem \ref{th1} also holds true for the spin operators in the Ruijsenaars form (\ref{q20}). The freezing trick provides long-range spin chain Hamiltonians, which coincide with those obtained from the Macdonald form in Theorem \ref{th1}.
\end{theor}

In Section \ref{sect6} the limits to trigonometric models were studied.
\begin{theor}\label{th3}
The Theorem \ref{th1} holds true for the spin operators (\ref{q10}) constructed by means of trigonometric $R$-matrices (\ref{q200})-(\ref{q203}).
\end{theor}
In the particular case of ${\rm U}_q({\widehat {\rm  gl}_2})$ $R$-matrix (the standard 6-vertex XXZ case) the q-deformed Haldane-Shastry model is reproduced in the form presented in \cite{Lam,LPS}. Higher rank analogues based on ${\rm U}_q({\widehat {\rm  gl}_M})$ $R$-matrix are included into the family of integrable long-range spin chains as well
(see Section 6 and Appendix C).
Finally, in Section \ref{sect7} we considered the non-relativistic limits $\hbar\rightarrow 0$. In this way we obtained two first elliptic commuting Hamiltonians (\ref{q66}), (\ref{q73}) obtained previously in \cite{SeZ}.

\section{Appendix}\label{sectA}
\subsection{A: Elliptic functions and $R$-matrix}
\def\theequation{A.\arabic{equation}}
\setcounter{equation}{0}

We use the odd ($\vth(-z)=-\vth(z)$) theta-function
\beq\label{a0963}\begin{array}{c} \displaystyle{
     \vartheta (z)=\vartheta (z|\tau) = -\sum_{k\in \mathbb{Z}} \exp \left( \pi \imath \tau (k + \frac{1}{2})^2 + 2\pi \imath (z + \frac{1}{2}) (k + \frac{1}{2}) \right)\,,\quad {\rm Im}(\tau)>0\,,
}\end{array}\eq
which has simple zero at $z=0$.
The  Kronecker elliptic function \cite{Weil} is as follows:
\beq\displaystyle{
\label{a0962}
    \phi(z, u) =
            \frac{\vartheta'(0) \vartheta (z + u)}{\vartheta (z) \vartheta (u)}\,.
}\eq
 It has simple pole at $z=0$ and the following local expansion (near $z=0$):
\beq\label{serphi}
\begin{array}{c} \displaystyle{
    \phi(z, u) = \frac{1}{z} + E_1 (u) + z\,\frac{E^2_1(u) - \wp(u)}{2} + O(z^2)\,,
}\end{array}\eq
where
\beq\label{serE}
\begin{array}{c}
\displaystyle{
    E_1(z) = \partial_z \ln \vartheta(z) =\zeta(z)+ \frac{z}{3} \frac{\vartheta'''(0) }{\vartheta'(0)}=\frac{1}{z} + \frac{z}{3} \frac{\vartheta'''(0) }{\vartheta'(0)} + O(z^3)\,.
}\end{array}
\eq
Here we used $\wp(z)$ and $\zeta(z)$. These are the Weierstrass $\wp$- and $\zeta$-functions.
The quasi-periodic behaviour on the lattice of periods $\Gamma=\mZ\oplus\mZ\tau$ (of elliptic curve $\mC/\Gamma$) for
theta function
 \beq\label{a0961}
  \begin{array}{l}
  \displaystyle{
 \vth(z+1)=-\vth(z)\,,\qquad \vth(z+\tau)=-e^{-\pi\imath\tau-2\pi\imath z}\vth(z)\,,
 }
 \end{array}
 \eq
 yields
 \beq\label{a096}
  \begin{array}{l}
  \displaystyle{
 \phi(z+1,u)= \phi(z,u)\,,\qquad  \phi(z+\tau,u)= e^{-2\pi\imath u}\phi(z,u)\,.
 }
 \end{array}
 \eq
 The Kronecker function (\ref{a0962}) satisfies the addition formula
%
\beq \begin{array}{c} \label{Fay} \displaystyle{
    \phi(z_1, u_1) \phi(z_2, u_2) = \phi(z_1, u_1 + u_2) \phi(z_2 - z_1, u_2) + \phi(z_2, u_1 + u_2) \phi(z_1 - z_2, u_1)\,,
}\end{array}\eq
the identity
\beq\begin{array}{c} \displaystyle{
\label{a0964}
    \phi(z, u) \phi(z, -u) = \wp (z) - \wp (u)
}\end{array}\eq
and relation
\beq\begin{array}{c} \displaystyle{
\label{a0965}
    \p_u\phi(x, u)=\phi(x,u)(E_1(x+u)-E_1(u))\,.
}\end{array}\eq
The latter directly follows from the definition (\ref{a0962}).

For the elliptic ${\rm GL}_M$ Baxter-Belavin $R$-matrix \cite{Baxter} the following set of $M^2$ functions is used:
 \beq\label{a08}
 \begin{array}{c}
  \displaystyle{
 \vf_a(z,\om_a+\hbar)=\exp(2\pi\imath\frac{a_2z}{M})\,\phi(z,\om_a+\hbar)\,,\quad
 \om_a=\frac{a_1+a_2\tau}{M}\,,
 }
 \end{array}
 \eq
where  $a=(a_1, a_2)\in\mZ_M\times\mZ_M$.
We also need a special matrix basis in $\MatM$:
\beq\label{a971}\begin{array}{c} \displaystyle{
        T_\al = \exp \left( \al_1 \al_2 \frac{\pi \imath}{M} \right) Q^{\al_1} \Lambda^{\al_2}, \quad \al = (\al_1, \al_2)\in \mZ_M \times \mZ_M\,,
}\end{array}\eq
where $Q,\Lambda\in\MatM$ are as follows:
\beq\label{a041}
 \begin{array}{c}
  \displaystyle{
 Q_{kl}=\delta_{kl}\exp(\frac{2\pi
 \imath}{{ M}}k)\,,\ \ \
 \Lambda_{kl}=\delta_{k-l+1=0\,{\hbox{\tiny{mod}}}\,
 { M}}\,,\quad Q^{ M}=\Lambda^{ M}=1_M
 }
 \end{array}
 \eq
 %
 %
 Finally, the Baxter-Belavin elliptic $R$-matrix is defined as
\begin{equation}\label{BB}
\begin{array}{c}
    \displaystyle{
    R^{\hbar}_{12} (x) = \frac{1}{M}
    \sum_\al \varphi_\al (x, \frac{\hbar}{M} + \om_\al) T_\al \otimes T_{-\al}}\in\MatM^{\otimes 2}\,.
\end{array}
\end{equation}
 In the classical limit $\hbar\rightarrow 0$
 \beq\label{r052}
 \begin{array}{c}
  \displaystyle{
 R_{12}^\hbar(z)=\hbar^{-1}1_M\otimes 1_M+r_{12}(z)+\hbar\, m_{12}(z)+O(\hbar^2)\,,
}
 \end{array}
 \eq
 it provides the classical Belavin-Drinfeld elliptic $r$-matrix \cite{BD} %
 \beq\label{r053}
 \begin{array}{c}
    \displaystyle{
    r_{12} (z) = \frac{1}{M}\, E_1(z) 1_M \otimes 1_M + \frac{1}{M}
    \sum_{\al \neq 0} \varphi_\al (z, \om_\al) T_\al \otimes T_{-\al}\,,
    }
\end{array}
\eq
 which satisfies the classical Yang-Baxter equation:
 \beq\label{r054}
 \begin{array}{c}
    \displaystyle{
    [r_{12},r_{23}]+[r_{12},r_{13}]+[r_{13},r_{23}]=0\,,
    \qquad r_{ij}=r_{ij}(z_i-z_j)\,.
    }
\end{array}
\eq
The residue of $R_{12}^\hbar(z)$ (and $r_{12}(z)$) at $z=0$ is the permutation operator
\begin{equation}\label{P12}
\begin{array}{c}
    \displaystyle{
    \res\limits_{z=0}R_{12}^\hbar(z)=\res\limits_{z=0}r_{12}(z)= P_{12}=\sum\limits_{k,l=1}^M e_{kl}\otimes e_{lk}=\frac{1}{M}\sum\limits_{\al\in\,\mZ_M\times\mZ_M}T_\al\otimes T_{-\al}\,,
    }
\end{array}
\end{equation}
where $\{e_{kl}\}$ is the standard matrix basis in $\MatM$. For any $a,b\in\mC^M$
we have $P_{12}(a\otimes b)=(b\otimes a)$,
and for any $A,B\in\MatM$: $P_{12}(A\otimes B)=(B\otimes A)P_{12}$. Also, $P_{12}^2=1_{M^2}$.

From the obvious property $\phi(z,u)=-\phi(-z,-u)$ one easily gets
the skew-symmetry property of (\ref{BB}):
 \beq\label{r055}
 \begin{array}{c}
    \displaystyle{
    R_{12}^{-\hbar}(-z)=-R_{21}^\hbar(z)\,.
    }
\end{array}
\eq
Then plugging into (\ref{r055}) the classical limit expansion (\ref{r052}) we obtain the following parity properties:
 \beq\label{r056}
 \begin{array}{c}
    \displaystyle{
    r_{12}(z)=-r_{21}(-z)\,,\qquad
    m_{12}(z)=m_{21}(-z)\,,\qquad
    \p r_{12}(z)=\p r_{21}(-z)\,,
    }
\end{array}
\eq
where $\p r_{12}(z)=\p_z r_{12}(z)$.

Presented here is a shorten version of the Appendix from  \cite{MZ}. See also the Appendix from \cite{ZZ}, where
different form of the elliptic $R$-matrix are briefly reviewed.

\subsection{B: Hamiltonians for $N=4$}
\def\theequation{B.\arabic{equation}}
\setcounter{equation}{0}

Here we use notations (\ref{s4601}), (\ref{s4602}) and $x_{ij}=x_i-x_j$. The Hamiltonians (\ref{s431}) for $N=4$ sites obtained through
(\ref{s39}) are as follows:
  \beq\label{ex03}
 \begin{array}{l}
H_1  =
 \displaystyle{
 \phi(x_{12})\phi(x_{32})\phi(x_{42})\bar{R}_{12}\bar{F}_{21}+
 }
 \displaystyle{
 \phi(x_{13})\phi(x_{23})\phi(x_{43})\bar{R}_{23}\bar{F}_{32}+
}
\\ \ \\
 \displaystyle{
 +\phi(x_{13})\phi(x_{23})\phi(x_{43})\bar{R}_{23}\bar{R}_{13}\bar{F}_{31}\bar{R}_{32}+
}
\displaystyle{
\phi(x_{14})\phi(x_{24})\phi(x_{34})\bar{R}_{34}\bar{F}_{43}+
}
\\ \ \\
\displaystyle{
+\phi(x_{14})\phi(x_{24})\phi(x_{34})\bar{R}_{34}\bar{R}_{24}\bar{F}_{42}\bar{R}_{43}+
}
\displaystyle{
\phi(x_{14})\phi(x_{24})\phi(x_{34})\bar{R}_{34}\bar{R}_{24}\bar{R}_{14}\bar{F}_{41}\bar{R}_{42}\bar{R}_{43}\,,\quad\quad\
}
 \end{array}
 \eq
 \beq\label{ex04}
 \begin{array}{l}
H_2  =  \phi(x_{21})\phi(x_{23})\phi(x_{41})\phi(x_{43})\bar{R}_{23}\bar{F}_{32}+
 \phi(x_{21})\phi(x_{24})\phi(x_{31})\phi(x_{34})\bar{R}_{34}\bar{F}_{43}+
 \\ \ \\
 +\phi(x_{21})\phi(x_{24})\phi(x_{31})\phi(x_{34})\bar{R}_{34}\bar{R}_{24}\bar{F}_{42}\bar{R}_{43}+
\phi(x_{12})\phi(x_{13})\phi(x_{42})\phi(x_{43})\bar{R}_{12}\bar{F}_{21}+
 \\ \ \\
+\phi(x_{12})\phi(x_{13})\phi(x_{42})\phi(x_{43})\bar{R}_{12}\bar{R}_{13}\bar{F}_{31}\bar{R}_{21}+
\phi(x_{12})\phi(x_{14})\phi(x_{32})\phi(x_{34})\bar{R}_{12}\bar{F}_{21}+
 \\ \ \\
+\phi(x_{12})\phi(x_{14})\phi(x_{32})\phi(x_{34})\bar{R}_{34}\bar{F}_{43}+
\phi(x_{12})\phi(x_{14})\phi(x_{32})\phi(x_{34})\bar{R}_{12}\bar{R}_{34}\bar{R}_{14}\bar{F}_{41}\bar{R}_{43}\bar{R}_{21}+
  \\ \ \\
+\phi(x_{13})\phi(x_{14})\phi(x_{23})\phi(x_{24})\bar{R}_{23}\bar{F}_{32}+
\phi(x_{13})\phi(x_{14})\phi(x_{23})\phi(x_{24})\bar{R}_{23}\bar{R}_{13}\bar{F}_{31}\bar{R}_{32}+
 \\ \ \\
+\phi(x_{13})\phi(x_{14})\phi(x_{23})\phi(x_{24})\bar{R}_{23}\bar{R}_{24}\bar{F}_{42}\bar{R}_{32}+
\\ \ \\
+\phi(x_{13})\phi(x_{14})\phi(x_{23})\phi(x_{24})\bar{R}_{23}\bar{R}_{13}\bar{R}_{24}\bar{R}_{14}\bar{F}_{41}\bar{R}_{42}\bar{R}_{31}\bar{R}_{32}\,,
 \end{array}
 \eq
 \beq\label{ex05}
 \begin{array}{l}
 H_3=\phi(x_{31})\phi(x_{32})\phi(x_{34})\bar{R}_{34}\bar{F}_{43}
   +\phi(x_{21})\phi(x_{23})\phi(x_{24})\bar{R}_{23}\bar{F}_{32}+
  \\ \ \\
 +\phi(x_{21})\phi(x_{23})\phi(x_{24})\bar{R}_{23}\bar{R}_{24}\bar{F}_{42}\bar{R}_{32}+
 \phi(x_{12})\phi(x_{13})\phi(x_{14})\bar{R}_{12}\bar{F}_{21}+
  \\ \ \\
 +\phi(x_{12})\phi(x_{13})\phi(x_{14})\bar{R}_{12}\bar{R}_{13}\bar{F}_{31}\bar{R}_{21}+
 \phi(x_{12})\phi(x_{13})\phi(x_{14})\bar{R}_{12}\bar{R}_{13}\bar{R}_{14}\bar{F}_{41}\bar{R}_{31}\bar{R}_{21}\,.\qquad\qquad
 \end{array}
 \eq
 The modified version (\ref{s461}), (\ref{s472}) of these Hamiltonians:
  %
  $$
 \begin{array}{lll}
{\bf H}_1 & = &
 \displaystyle{
\bar{R}_{12}(x_1-x_2)\bar{F}_{21}(x_2-x_1)
+\bar{R}_{23}(x_2-x_3)\bar{F}_{32}(x_3-x_2)+\bar{R}_{34}(x_3-x_4)\bar{F}_{43}(x_4-x_3)+
 }
  \end{array}
 $$
  \beq\label{ex03a}
 \begin{array}{lll}
\phantom{{\bf H}_1} & \phantom{=} &
 \displaystyle{
 +\bar{R}_{23}(x_2-x_3)\bar{R}_{13}(x_1-x_3)\bar{F}_{31}(x_3-x_1)\bar{R}_{32}(x_3-x_2)+
 }
\\ \ \\
 & &
 \displaystyle{
 +\bar{R}_{34}(x_3-x_4)\bar{R}_{24}(x_2-x_4)\bar{F}_{42}(x_4-x_2)\bar{R}_{43}(x_4-x_3)+
}
\\ \ \\
& &
\displaystyle{
+\bar{R}_{34}(x_3-x_4)\bar{R}_{24}(x_2-x_4)\bar{R}_{14}(x_1-x_4)\bar{F}_{41}(x_4-x_1)\bar{R}_{42}(x_4-x_2)\bar{R}_{43}(x_4-x_3)\,,
}
 \end{array}
 \eq

 \beq\label{ex04a}
 \begin{array}{lll}
{\bf H}_2 & = &\left( \frac{1}{\wp(\hbar)-\wp(\frac{1}{2})}+\frac{1}{\wp(\hbar)-\wp(\frac{1}{4})}\right)
\Big(\bar{R}_{12}(x_1-x_2)\bar{F}_{21}(x_2-x_1)+
\\ \ \\
 & &\qquad\qquad\qquad\qquad\qquad+\bar{R}_{23}(x_2-x_3)\bar{F}_{32}(x_3-x_2)+\bar{R}_{34}(x_3-x_4)\bar{F}_{43}(x_4-x_3)\Big)+
 \\ \ \\
& &+\frac{1}{\wp(\hbar)-\wp(\frac{1}{4})}\bar{R}_{12}(x_1-x_2)\bar{R}_{13}(x_1-x_3)\bar{F}_{31}(x_3-x_1)\bar{R}_{21}(x_2-x_1)+
\\ \ \\
 & &+\frac{1}{\wp(\hbar)-\wp(\frac{1}{4})}\bar{R}_{34}(x_3-x_4)\bar{R}_{24}(x_2-x_4)\bar{F}_{42}(x_4-x_2)\bar{R}_{43}(x_4-x_3)+
\\ \ \\
 & & +\frac{1}{\wp(\hbar)-\wp(\frac{1}{4})}\bar{R}_{23}(x_2-x_3)\bar{R}_{13}(x_1-x_3)\bar{F}_{31}(x_3-x_1)\bar{R}_{32}(x_3-x_2)+\qquad\qquad\qquad\qquad\qquad\
 \\ \ \\
& &+\frac{1}{\wp(\hbar)-\wp(\frac{1}{4})}\bar{R}_{23}(x_2-x_3)\bar{R}_{24}(x_2-x_4)\bar{F}_{42}(x_4-x_2)\bar{R}_{32}(x_3-x_2)+
  \end{array}
 \eq
  $$
 \begin{array}{lll}
 \phantom{{\bf H}_2} & \phantom{=} &
  +\frac{1}{\wp(\hbar)-\wp(\frac{1}{2})}\bar{R}_{12}(x_1-x_2)\bar{R}_{34}(x_3-x_4)\bar{R}_{14}(x_1-x_4)\bar{F}_{41}(x_4-x_1)\times
  \qquad\qquad\qquad\qquad\qquad\
 \\ \ \\
 & & \qquad\qquad\qquad\qquad\qquad\qquad\qquad\quad\times\bar{R}_{43}(x_4-x_3)\bar{R}_{21}(x_2-x_1)+
  \\ \ \\
& &+\frac{1}{\wp(\hbar)-\wp(\frac{1}{4})}\bar{R}_{23}(x_2-x_3)\bar{R}_{13}(x_1-x_3)\bar{R}_{24}(x_2-x_4)\bar{R}_{14}(x_1-x_4)
\times
\\ \ \\
& & \qquad\qquad \times\bar{F}_{41}(x_4-x_1)\bar{R}_{42}(x_4-x_2)\bar{R}_{31}(x_3-x_1)\bar{R}_{32}(x_3-x_2)\,.
 \end{array}
 $$
Redefine the third Hamiltonian as $\displaystyle{\bf H}_3=\frac{1}{\phi(\frac{1}{4})\phi(\frac{2}{4})\phi(\frac{3}{4})}H_3$, the expression is the following:
 \beq\label{ex05a}
 \begin{array}{lll}
\displaystyle {\bf H}_3 & = &\bar{R}_{34}(x_3-x_4)\bar{F}_{43}(x_4-x_3)
+\bar{R}_{23}(x_2-x_3)\bar{F}_{32}(x_3-x_2)+
\\ \ \\
& &+\bar{R}_{12}(x_1-x_2)\bar{F}_{21}(x_2-x_1)+\bar{R}_{23}(x_2-x_3)\bar{R}_{24}(x_2-x_4)\bar{F}_{42}(x_4-x_2)\bar{R}_{32}(x_3-x_2)+
  \\ \ \\
& & +\bar{R}_{12}(x_1-x_2)\bar{R}_{13}(x_1-x_3)\bar{F}_{31}(x_3-x_1)\bar{R}_{21}(x_2-x_1)+
  \\ \ \\
& &+\bar{R}_{12}(x_1-x_2)\bar{R}_{13}(x_1-x_3)\bar{R}_{14}(x_1-x_4)\bar{F}_{41}(x_4-x_1)\bar{R}_{31}(x_3-x_1)\bar{R}_{21}(x_2-x_1)\,.
 \end{array}
 \eq

%
\subsection{C: Trigonometric ${\rm U}_q({\widehat {\rm  gl}_M})$ $R$-matrix}
\def\theequation{C.\arabic{equation}}
\setcounter{equation}{0}

Here we explain that the spin operators (\ref{q10}) with $R$-matrix (\ref{q270}) commute
and the freezing trick provides commutative set of operators based on this $R$-matrix.

As was shown in \cite{Uglov} (see also \cite{Lam,LPS})
the Polychronakos freezing trick works for spin operators with $R$-matrix (\ref{q270}) in the case $M=2$. In fact, since the relations (\ref{q39}) are independent of $M$, the freezing trick holds true for all $M>2$ as well. In this
way the higher rank q-deformed Haldane-Shastry model is included into the family of models under consideration.

Let us now focus on the proof of commutativity of spin operators (\ref{q10}). It was shown in \cite{MZ}
that the commutativity is equivalent to a set of identities (Theorem 1,\cite{MZ}), and the
elliptic $R$-matrix satisfy these identities (Theorem 2, \cite{MZ}). The proof of Theorem 1 does use explicit form of $R$-matrix. We argue below
that the proof of Theorem 2 from \cite{MZ} for $k>1$ can be naturally applied to $R$-matrix (\ref{q270}).
For identities with $k=1$ the associative Yang-Baxter equation was used in \cite{MZ}. The latter equation is valid
for (\ref{q270}) with $M=2$, but is not valid for $M>2$. For this reason we also give a proof of identities for $k=1$ without usage of the associative Yang-Baxter equation.

\paragraph{Proof of $R$-matrix identities.} Consider the XXZ ${\rm U}_q({\widehat {\rm  gl}_M})$ $R$-matrix (\ref{q270}) and rewrite it
using notation (\ref{q31}):
  \beq\label{q2701}
   \begin{array}{c}
   \displaystyle{
  R^{\rm trig}_{12}(u)=
   \pi\imath\Big(\frac{u+1}{u-1}+\frac{t+1}{t-1}\Big)\sum\limits_{i=1}^M e_{ii}\otimes e_{ii}+
  }
  \\ \ \\
   \displaystyle{
 +2\pi\imath \frac{t^{1/2}}{t-1}\sum\limits_{i\neq j}^M e_{ii}\otimes e_{jj}+
 2\pi\imath\sum\limits_{i< j}^M
 \Big( e_{ij}\otimes e_{ji}\,\frac{u}{u-1}+e_{ji}\otimes
 e_{ij}\,\frac{1}{u-1}\Big)\,.
  }
  \end{array}
  \eq
  This $R$-matrix is normalized as given in (\ref{q03}), where the function $\phi$ is replaced by its trigonometric version $\phi^{\rm trig}(z)=a(u)$
  (\ref{q31a}). Introduce also notation
  \beq\label{q2702}
   \begin{array}{c}
   \displaystyle{
  x=e^{2\pi\imath\eta}\,,
  }
  \end{array}
  \eq
  where $\eta$ is the parameter entering the shift operator (\ref{p_i}).

  The left hand side of an $R$-matrix
  identity (denote it as $\mF$) is a sum of products of $R$-matrices $R_{ij}^\hbar(u_i/u_j)$ and shifted $R$-matrices $R_{ij}^\hbar(u_i/(xu_j))$.
  It is easy to see from (\ref{q2701}) that any shifted $R$-matrix being considered as function of $x$ is represented in the following form:
  \beq\label{q2703}
   \begin{array}{c}
   \displaystyle{
  R^{\rm trig}_{12}(u/x)=\frac{A_{12}(u,t)}{x-u}+B_{12}(u,t)\,.
  }
  \end{array}
  \eq
  Main idea in the proof of Theorem 2 from \cite{MZ} was to show that  $\mF$ is independent of $\eta$ (or $x$). Then one can put $\eta=0$ (or $x=1$). In this case due to unitarity condition (\ref{q03}) $\mF$ boils down to the scalar case, for which $\mF|_{x=1}=0$ is known to be valid. Due to (\ref{q2703}) any product of shifted and unshifted $R$-matrices (with distinct arguments)
  has the form:
  \beq\label{q2704}
   \begin{array}{c}
   \displaystyle{
  \mF(x)=\sum\limits_{a=1}^{N^2-N}\frac{C^a}{x-u_a}+D\,,
  }
  \end{array}
  \eq
  where index $a$ enumerates all possible arguments $u_i/u_j$ ($i,j=1,...,N$, $i\neq j$), and $C^a$ - are some
  operator valued coefficients. Notice that the coefficients $C^a$ and $D$ are independent of $x$.   The absence of poles in $\mF$ at $\eta=z_i-z_j$ ($x=u_i/u_j$) was proved in
   Theorem 2 from \cite{MZ}. Although the proof was performed for elliptic $R$-matrix, it holds true (to be precise, the part of the proof which states that the poles at $\eta=z_i-z_j$ are absent holds true) for any
   $R$-matrix with the property $\res\limits_{u=1}R_{12}(u)=P_{12}$, and it is true for  (\ref{q2701}).
   In this way we conclude that $C^a=0$ for all $a=1,...,N^2-N$. Then $\mF(x)=D$, which is independent of $x$.
   Therefore, one can set $x=1$, and $\mF=\mF|_{x=1}=0$.

Let us now prove $R$-matrix identity for $k=1$ without usage of the associative Yang-Baxter equation.
In $k=1$ case the identity is as follows:
 \beq\label{q2705}
  \begin{array}{c}
 \mF= \displaystyle{
 \sum\limits_{k=1}^N\overrightarrow{\prod\limits_{i=k+1}^N} R_{ki}^{\rm trig}(u_k/u_i)
 \overleftarrow{\prod\limits_{j:j\neq k}^N} R_{jk}^{\rm trig}(u_j/(x u_k))
 \overrightarrow{\prod\limits_{l=1}^{k-1}} R_{kl}^{\rm trig}(u_k/u_l)
 -
 }
 \\ \ \\
  \displaystyle{
 -\sum\limits_{k=1}^N\overleftarrow{\prod\limits_{l=1}^{k-1}} R_{lk}^{\rm trig}(u_l/u_k)
 \overrightarrow{\prod\limits_{j:j\neq k}^N} R_{kj}^{\rm trig}(u_k/(x u_j))
 \overleftarrow{\prod\limits_{i=k+1}^N} R_{ik}^{\rm trig}(u_i/u_k)=0
 \,,
 }
 \end{array}
 \eq
or equivalently,
 \beq\label{aq2706}
  \begin{array}{c}
 \mF= \displaystyle{
 \sum\limits_{k=1}^N R_{k,k+1}\dots R_{k,N}
 \cdot R_{N,k}^{-}\dots R_{k+1,k}^{-}R_{k-1,k}^{-}\dots R_{1,k}^{-}
 \cdot R_{k,1}\dots R_{k,k-1}
 -
 }
 \\ \ \\
  \displaystyle{
 -\sum\limits_{k=1}^N R_{k-1,k}\dots R_{1,k}
 \cdot R_{k,1}^{-}\dots R_{k,k-1}^{-}R_{k,k+1}^{-}\dots R_{k,N}^{-}
 \cdot R_{N,k}\dots R_{k+1,k}=0\,,
 }
 \end{array}
 \eq
where notations $R_{ij}=R_{ij}(u_i/u_j)$ and $R^{-}_{ij}=R_{ij}(u_i/(x u_j))$ are used for shortness.

Calculate residue of $\mF$ at $x=u_i/u_j$ for $i<j$:
\beq\label{aq2707}
  \begin{array}{c}
    \displaystyle{
 \res\limits_{x=u_i/u_j}\mF=
  }
 \\ \ \\
  \displaystyle{
\res\limits_{x=u_i/u_j} R_{j,j+1}\dots R_{j,N}
 \cdot R_{N,j}^{-}\dots R_{j+1,j}^{-}R_{j-1,j}^{-}\dots R_{i+1,j}^{-}\underline{R_{i,j}^{-}}R_{i-1,j}^{-}\dots R_{1,j}^{-}
 \cdot R_{j,1}\dots R_{j,j-1}
 }
 \\ \ \\
  \displaystyle{
 -\res\limits_{x=u_i/u_j} R_{i-1,i}\dots R_{1,i}
  R_{i,1}^{-}\dots R_{i,i-1}^{-}R_{i,i+1}^{-}\dots R_{i,j-1}^{-}\underline{R_{i,j}^{-}}R_{i,j+1}^{-}\dots R_{i,N}^{-}
 \cdot R_{N,i}\dots R_{i+1,i}\,.
 }
 \end{array}
 \eq
 The underlined $R$-matrix is the one which has pole at $x=u_i/u_j$. Using the property $\res\limits_{u=1}R_{ij}(u)=P_{ij}$ and moving $P_{ij}$ to the right one obtains:
 \beq\label{aq2708}\notag
  \begin{array}{l}
= \displaystyle{
R_{j,j+1}\dots R_{j,N}
 \cdot R_{N,j}^{-}\dots R_{j+1,j}^{-}R_{j-1,j}^{-}\dots R_{i+1,j}^{-}\cdot\underbrace{ R_{i-1,i}\dots R_{1,i}}
 \cdot \underbrace{R_{i,1}^{-}\dots  R_{i,i-1}^-}\times
  }
 \\ \ \\
 \displaystyle{
 \ \times R_{ij}(u_j/u_i)R_{i,i+1}^- \dots R_{i,j-1}^- P_{ij}-
 }
 \\ \ \\
 \displaystyle{
 -R_{i-1,i}\dots R_{1,i}\cdot
  R_{i,1}^{-}\dots R_{i,i-1}^{-}R_{i,i+1}^{-}\dots R_{i,j-1}^{-}\cdot \underbrace{R_{j,j+1}\dots R_{j,N}}
 \cdot \underbrace{R_{N,j}^-\dots R_{j+1,j}^-}\times
   }
 \\ \ \\
 \displaystyle{
 \ \times R_{ij}(u_j/u_i)R_{j-1,j}^-\dots R_{i+1,j}^-P_{ij}\,
 }
 \end{array}
 \eq
Here we use the same short notations $R_{ij}$ and $R_{ij}^-$ assuming $x=u_i/u_j$, so that $R_{ik}^-:=R_{ik}(u_j/u_k)$.
The underbraced factors can be moved to the left. This provides the common factor:
\beq\label{aq2709}
  \begin{array}{c}
\res\limits_{x=u_i/u_j}\mF= \displaystyle{
R_{j,j+1}\dots R_{j,N}\cdot R_{i-1,i}\dots R_{1,i}\cdot R_{i,1}^{-}\dots  R_{i,i-1}^-\cdot R_{N,j}^{-}\dots R_{j+1,j}^{-}\times
}
\\ \ \\
\displaystyle{
\times\Big(
 R_{j-1,j}^{-}\dots R_{i+1,j}^{-}
 \cdot R_{ij}(u_j/u_i)\cdot R_{i,i+1}^- \dots R_{i,j-1}^-
 \qquad\qquad\qquad\qquad\qquad
}
\\ \ \\
\displaystyle{
\qquad\qquad\qquad\qquad\qquad
-R_{i,i+1}^{-}\dots R_{i,j-1}^{-}\cdot
 R_{ij}(u_j/u_i)\cdot R_{j-1,j}^-\dots R_{i+1,j}^-\Big)P_{ij}\,.
 }
 \end{array}
 \eq
 The last one step is to prove that the expression inside the brackets in (\ref{aq2709}) equals zero, or
 equivalently
 \beq\label{aq2710}
  \begin{array}{c}
\displaystyle{
 R_{j-1,j}^{-}\dots R_{i+1,j}^{-}
 \cdot R_{ij}(u_j/u_i)\cdot R_{i,i+1}^- \dots R_{i,j-1}^-
=
 \qquad\qquad\qquad\qquad\qquad
}
\\ \ \\
\displaystyle{
 \qquad\qquad\qquad\qquad\qquad
=
R_{i,i+1}^{-}\dots R_{i,j-1}^{-}\cdot
 R_{ij}(u_j/u_i)\cdot R_{j-1,j}^-\dots R_{i+1,j}^-\,.
 }
 \end{array}
 \eq
 By applying the Yang-Baxter equation $$R_{k,j}(u_k/u_i)
  \cdot R_{i,j}(u_j/u_i)\cdot R_{i,k}(u_j/u_k)= R_{i,k}(u_j/u_k)\cdot R_{i,j}(u_j/u_i)R_{k,j}(u_k/u_i)$$
   to the left hand side of (\ref{aq2710}) several times one obtains the right hand side of (\ref{aq2710}).


\subsection*{Acknowledgments}


We are grateful to J. Lamers and A.P. Polychronakos for useful comments and remarks.

This work was supported by the Russian Science Foundation under grant no. 19-11-00062,\\ https://rscf.ru/en/project/19-11-00062/ .


\begin{small}

\end{small}

\end{document}